\documentclass{article}

\usepackage{arxiv}

\usepackage[utf8]{inputenc} 
\usepackage[T1]{fontenc}    
\usepackage{hyperref}       
\usepackage{url}            
\usepackage{booktabs}       
\usepackage{amsfonts}       
\usepackage{nicefrac}       
\usepackage{microtype}      
\usepackage{lipsum}

\usepackage{morefloats}
\usepackage{color}
\usepackage{multirow}
\usepackage{latexsym}
\usepackage{amsmath,amssymb,amsthm,graphicx}
\usepackage{dsfont}
\usepackage{enumitem}
\usepackage{hyperref}

\newtheoremstyle{thm}
{9pt}
{9pt}
{\itshape}
{}
{\bfseries}
{.}
{ }
{}
\theoremstyle{thm}

\newtheorem{theorem}{Theorem}[section]
\newtheorem{lemma}[theorem]{Lemma}
\newtheorem{corollary}[theorem]{Corollary}
\newtheorem{prop}[theorem]{Proposition}
\newcommand{\vertk}{\stackrel{{\cal D}}{\longrightarrow}}
\newcommand{\fsk}{\stackrel{{\rm a.s.}}{\longrightarrow}}

\newtheoremstyle{def}
{9pt}
{9pt}
{}
{}
{\bfseries}
{.}
{ }
{}
\theoremstyle{def}

\newtheorem{remark}[theorem]{Remark}
\newtheorem{example}[theorem]{Example}

\newcommand{\R}{\mathbb{R}} 
\newcommand{\E}{\mathbb{E}} 
\newcommand{\PP}{\mathbb{P}} 
\newcommand{\HH}{\mathbb{H}} 

\newcommand{\fse}{\overset{\emph{a.s.}}{\longrightarrow}}

\renewcommand{\footnoterule}{%
	\kern -3.5pt
	\hrule width \textwidth height 1pt
	\kern 3.5pt
}

\makeatletter
\def\blfootnote{\xdef\@thefnmark{}\@footnotetext}
\makeatother

\title{Testing multivariate normality by zeros of the harmonic oscillator in characteristic function spaces}

\author{Philip D\"orr\\
Institute of Stochastics, \\
Karlsruhe Institute of Technology (KIT), \\
Englerstr. 2, D-76133 Karlsruhe. \\
\And  Bruno Ebner\\
 Institute of Stochastics, \\
Karlsruhe Institute of Technology (KIT), \\
Englerstr. 2, D-76133 Karlsruhe. \\
\texttt{Bruno.Ebner@kit.edu}\\
\And
Norbert Henze\\
Institute of Stochastics, \\
Karlsruhe Institute of Technology (KIT), \\
Englerstr. 2, D-76133 Karlsruhe. \\
\texttt{Norbert.Henze@kit.edu}\\
}

\begin{document}

\date{\today}
\maketitle


\begin{abstract}
We study a novel class of affine invariant and consistent tests for normality	in any dimension.
The tests are based on a characterization of the standard $d$-variate normal distribution as the unique solution
of an initial value problem of a partial differential equation motivated by the harmonic oscillator,
which is a special case of a Schr\"odinger operator.  We derive the asymptotic distribution of the
test statistics under the hypothesis of normality as well as under fixed and contiguous alternatives.
The tests are consistent against general alternatives, exhibit strong power performance for finite samples,
and they are applied to a classical data set due to  R.A. Fisher. The results can also be used for a neighborhood-of-model validation
procedure.
\end{abstract}

\section{Introduction}\label{sec:Intro}
The multivariate normal distribution plays a key role in classical and hence widely used procedures.
Serious statistical inference that involves the assumption of multivariate normality should therefore start with
 a test of fit to this model. There is a continuing interest in this testing problem, as evidenced by a multitude of papers.
 The proposed tests may be roughly classified as follows: \cite{A:2007,BH:1988,HW:1997,HZ:1990,P:2005,T:2009} consider
 tests based on the empirical characteristic function, while \cite{HJG:2019,HJM:2019,HV:2019} employ the
 empirical moment generating function.
A classical (and still popular) approach is to consider measures of multivariate skewness and kurtosis
(see, e.g., \cite{DH:2008,KTO:2007,MA:1973,MAR:1970,MAR:1974,MRS:1993}), as supposedly diagnostic tools
with regard to the kind of deviation from normality when this hypothesis has been rejected, but the deficiencies of those measures in this regard have been clearly demonstrated
(see, e.g., \cite{BH:1991,BH:1992,HEN:1994b,HEN:1994a,H:1997}). Other approaches involve generalizations of tests for univariate normality
\cite{AE:2009,KP:2018,S:2006}, the examination of nonlinearity of dependence \cite{CS:1978,E:2012}, canonical correlations \cite{T:2014}
and the notion of energy, see \cite{SR:2005}. For a survey of affine invariant tests for multivariate normality, see  \cite{H:2002}.
 Monte Carlo studies can be found in  \cite{FSN:2007,MM:2005,VPMV:2016}.

To be specific, let $X,X_1,\ldots,X_n,\ldots$ be a sequence of independent identically distributed (i.i.d.) $d$-dimensional random (column) vectors,
which are defined on some common probability space $(\Omega,{\cal A},\mathbb{P})$. Here, $d\ge1$ is a fixed integer, which means that
the univariate case is deliberately not excluded.  We write $\mathbb{P}^X$ for the distribution of $X$. The $d$-variate normal distribution with expectation
$\mu$ and nonsingular covariance matrix $\Sigma$ will be denoted by N$_d(\mu,\Sigma)$. Furthermore,
\begin{equation*}
\mathcal{N}_d=\{\text{N}_d(\mu,\Sigma):\mu\in\mathbb{R}^d,\,\Sigma\in\mathbb{R}^{d\times d}\; \mbox{positive definite}\}
\end{equation*}
stands for the family of nondegenerate $d$-variate normal distributions. A check of the assumption of multivariate normality means to test the hypothesis
\begin{equation}\label{H0}
H_0:\,\mathbb{P}^X\in\mathcal{N}_d,
\end{equation}
against general alternatives.

Writing I$_d$ for the unit matrix of order $d$, our novel idea for testing $H_0$ is to use a characterization of the Fourier transform
 of N$_d(0,\text{I}_d)$ as the unique solution of an initial value problem of a partial differential equation
motivated by the harmonic oscillator, which is a special case of a Schr\"{o}dinger operator.
The proposed test statistic is based on the squared norm of a functional of the empirical characteristic function in a suitably weighted $L^2$-space.
This statistic is close to zero under the hypothesis (\ref{H0}), and rejection will be for large values of the test statistic.

Let $L^2(\R^d)$ be the space of square integrable functions, equipped with the usual norm and scalar product $\langle\cdot,\cdot\rangle$.
Consider for sufficiently regular $f \in L^2(\R^d)$ the partial differential equation
\begin{equation}\label{PDE}
\left\{\begin{array}{l}\Delta f(x) = (\|x\|^2-d) f(x),\quad x\in\R^d,\\ f(0)=1.\end{array}\right.
\end{equation}
Here, $\Delta$ stands for the Laplace operator, and $\|\cdot\|$ denotes the Euclidean norm. Notice that we can rewrite (\ref{PDE}) as
$(-\Delta + \|x\|^2 - d) f(x)= 0$
or, equivalently, as
\begin{equation}\label{PDE2}
\sum_{j=1}^d\left(-\frac{\partial^2}{\partial x_j^2} + x_j^2 - 1\right) f(x) = 0,\quad x=(x_1,\ldots,x_d) \in\R^d.
\end{equation}
  The operator $-\Delta + \|x\|^2 - d$ is known as the harmonic oscillator, see \cite{G:2011}. In the univariate case,
   (\ref{PDE}) reduces to a fixed point problem (or, equivalently, to the problem of finding the eigenfunction that corresponds
   to the eigenvalue 1) of the Hermite operator, see equation (1.1.9) in \cite{T:1993}.
   The solution of this problem is the $0^{th}$ Hermite function, which coincides with the solution given in the following theorem.
\begin{theorem}\label{unique.sol}
The characteristic function
\begin{equation}\label{defpsit}
\psi(t)=\exp\left(-\frac{\|t\|^2}{2}\right),\quad t\in\R^d,
\end{equation}
of the $d$-variate standard normal distribution {\rm \textrm{N$_d(0,{\rm {\textrm{I}}}_d)$}}
is the unique solution of (\ref{PDE}).
\end{theorem}
\begin{proof}
Let $f$ be an arbitrary solution of (\ref{PDE}). Writing i for the imaginary unit, we introduce the creation and annihilation operators
$a_j=x_j+\text{i}p_j$ and $a_j^\star=x_j-\text{i}p_j$, $j=1,\ldots,d$,
where $p_j=-\text{i} \frac{\partial}{\partial x_j}$, $j=1,\ldots,d$. For each $j \in \{1,\ldots,d\}$ we have
\begin{eqnarray*}
a_j^\star a_j & = & \left(x_j-\frac{\partial}{\partial x_j}\right)\left(x_j+\frac{\partial}{\partial x_j}\right) =
 x_j^2+x_j\frac{\partial}{\partial x_j}-\frac{\partial}{\partial x_j}x_j-\frac{\partial^2}{\partial x_j^2}\\
& = & -\frac{\partial^2}{\partial x_j^2}+x_j^2-1.
\end{eqnarray*}
So we can rewrite (\ref{PDE2}) as
$\sum_{j=1}^da_j^\star a_j f = 0$,
which implies
\begin{equation*}
\bigg{\langle} f,\sum_{j=1}^da_j^\star a_j f\bigg{\rangle} \,
= \sum_{j=1}^d\langle f,a_j^\star a_j f \rangle\,=\sum_{j=1}^d\langle a_jf, a_j f\rangle\,=\sum_{j=1}^d\|a_j f\|^2=0
\end{equation*}
and thus $a_jf=0$ for each $j \in \{1,\ldots,d\}$. We therefore have $(x+\nabla)f=0$, where $\nabla$ denotes the gradient operator.
By the last statement and the product rule, it follows that
\begin{equation*}
\nabla \left(\exp\left(\frac{\|x\|^2}{2}\right) f\right)=x\exp\left(\frac{\|x\|^2}{2}\right)f-x\exp\left(\frac{\|x\|^2}{2}\right)f=0
\end{equation*}
which, in view of the condition $f(0)=1$, completes the proof.
\end{proof}
\begin{remark}
  The operator $H=-\Delta + \|x\|^2$ is the Hermite operator in $\R^d$, and $\psi$ is the product of the one-dimensional $0^{th}$ Hermite functions.
  Therefore, since $H\psi=d\psi$ (as we have shown in Theorem \ref{unique.sol}), $\psi$ is the eigenfunction associated with the eigenvalue $d$.
  For details on the Hermite operator in $\R^d$ and corresponding eigenfunctions, see p. 5 of \cite{T:1993}.
\end{remark}
In this paper, we study a family of affine invariant test statistics for $H_0$ that is based on the characterization given in Theorem \ref{unique.sol}.
Since the class  $\mathcal{N}_d$ is closed under full rank affine transformations, Theorem \ref{unique.sol} does not restrict the scope of the testing problem.
We make the tacit standing assumption that $\mathbb{P}^X$ is absolutely continuous with respect to Lebesgue measure, and that
 $n \ge d+1$.  Let $\overline{X}_n= n^{-1} \sum_{j=1}^nX_j$ denote the sample mean and $S_n= n^{-1} \sum_{j=1}^n(X_j-\overline{X}_n)(X_j-\overline{X}_n)^\top$
the sample covariance matrix of $X_1,\ldots,X_n$, where $x^\top$ means transposition of a column vector $x$.
The assumptions made above guarantee that $S_n$ is invertible almost surely, see \cite{EP:1973}.
The test statistic will be based on the so-called {\em scaled residuals}
\begin{equation}\label{defscaledr}
Y_{n,j}=S_n^{-1/2}(X_j-\overline{X}_n), \quad j =1,\ldots,n,
\end{equation}
which represent an empirical standardization of the data. Here, $S_n^{-1/2}$ is the unique symmetric
positive definite square root of $S_n^{-1}$. The test statistic will be based on the
empirical characteristic function
\begin{equation*}
\psi_n(t)=\frac1n\sum_{j=1}^n\exp({\rm{i}}t^\top Y_{n,j}),\quad t\in\R^d,
\end{equation*}
of $Y_{n,1},\ldots,Y_{n,n}$.  Notice that an application of the Laplace operator $\Delta$ to $\psi_n$ yields
\begin{equation*}
\Delta\psi_n(t)=-\frac1n\sum_{j=1}^n\|Y_{n,j}\|^2\exp({\rm{i}}t^\top Y_{n,j}),\quad t\in\R^d.
\end{equation*}
Motivated by (\ref{PDE}) and Theorem \ref{unique.sol}, we propose the weighted $L^2$-statistic
\begin{eqnarray*}
T_{n,a}&=&n\int\left|\Delta\psi_n(t)-\Delta \psi(t) \right|^2w_a(t)\, \mbox{d}t\\
&=&n\int \Big{|}\frac1n\sum_{j=1}^n\|Y_{n,j}\|^2\exp({\rm{i}}t^\top Y_{n,j})+(\|t\|^2-d)\exp\left(-\frac{\|t\|^2}{2}\right)\Big{|}^2w_a(t)\, \mbox{d}t,
\end{eqnarray*}
where
\begin{equation}\label{defwa}
w_a(t)= \exp(-a\|t\|^2), \quad t \in \R^d,
\end{equation}
and $a>0$ is a fixed constant. Moreover, $|z|$ is the modulus of a complex number $z$, and integration is, unless otherwise specified, over $\R^d$.
In principle, other weight functions
than  $w_a$ are conceivable in the definition of $T_{n,a}$. Since, for $c \in \R^d$ (see  \cite{HZ:1990}, p. 3601),
\begin{equation} \label{glint1}
\int\cos(t^\top c)\exp(-a\|t\|^2)\mbox{d}t=\left(\frac\pi{a}\right)^\frac{d}{2}\exp\left(-\frac{\|c\|^2}{4a}\right),
\end{equation}
as well as
\begin{eqnarray} \nonumber
& & \int(\|t\|^2\! -\! d)\cos(t^\top c)\exp\left(\! -\left(a\! +\! \frac12\right)\|t\|^2\! \right)\mbox{d}t  \phantom{hhhhhhhhhhhhhhhhhhhhhhhhhhh} \\ \label{glint2}
   &  &  \phantom{hhhhhhhhhhhhhhhh} = \  -\frac{\left(2\pi\right)^{\frac{d}{2}} \left(\|c\|^2+2da(2a+1)\right)}{(2a+1)^{2+\frac{d}{2}}}\exp\left(\! -\frac{\|c\|^2}{2(2a\! +\! 1)}\! \right),
\end{eqnarray}
\begin{equation}\label{glint3}
\int(\|t\|^2-d)^2\exp\left(-(a+1)\|t\|^2\right)\mbox{d}t   =  \frac{\pi^{\frac{d}2}}{(a+1)^{2+\frac{d}{2}}}\left(a(a+1)d^2+\frac{d(d+2)}{4}\right),
\end{equation}
an attractive feature of the choice of $w_a$ is that the test statistic takes the simple form
\begin{eqnarray}\label{tna1}
T_{n,a}&=&\left(\frac\pi{a}\right)^\frac{d}{2}\frac1n\sum_{j,k=1}^n\|Y_{n,j}\|^2\|Y_{n,k}\|^2\exp\left(-\frac1{4a}\|Y_{n,j}-Y_{n,k}\|^2\right)\\ \label{tna2}
&&-\, \frac{2(2\pi)^{\frac{d}2}}{(2a+1)^{2+\frac{d}{2}}}\sum_{j=1}^n\|Y_{n,j}\|^2\left(\|Y_{n,j}\|^2+2da(2a+1)\right)\exp\left(-\frac12\frac{\|Y_{n,j}\|^2}{2a+1}\right)\\ \label{tna3}
&&+\, n\frac{\pi^{\frac{d}2}}{(a+1)^{2+\frac{d}{2}}}\left(a(a+1)d^2+\frac{d(d+2)}{4}\right),
\end{eqnarray}
which is amenable to computational purposes. Moreover, $T_{n,a}$ depends only on the scalar products $Y_{n,i}^\top Y_{n,j} = (X_i-\overline{X}_n)^\top S_n^{-1}(X_j-\overline{X}_n)$,
where $i,j \in \{1,\ldots,n\}$. This shows that $T_{n,a}$ is invariant with respect to full rank affine transformations of $X_1,\ldots,X_n$. Moreover,
not even the square root $S_n^{-1/2}$ of $S_n^{-1}$ is needed.

The rest of the paper is organized as follows: In Section \ref{secinfty}, we show that, as the tuning parameter $a$ tends to infinity,
the test statistic $T_{n,a}$, after a suitable scaling, converges to a  certain measure of multivariate skewness. On the other hand, a time-honored measure of multivariate
kurtosis emerges as $a \to 0$. Section \ref{basicCLT} presents a basic Hilbert space central limit theorem, which proves beneficial for obtaining the
limit distribution of $T_{n,a}$ both under $H_0$ and under fixed alternatives to normality.
In Section \ref{secnull}, we derive the limit null distribution of $T_{n,a}$ as $n \to \infty$. Section \ref{seccontig} considers the behavior of $T_{n,a}$ with respect to contiguous
alternatives to $H_0$. In Section \ref{secfixed}, we show that the test for multivariate normality that rejects $H_0$ for large values of $T_{n,a}$ is consistent against general
alternatives. Moreover, the limit distribution of $T_{n,a}$ under a fixed alternative  distribution is seen to be normal. Since the
variance of this normal distribution can be estimated from the data, there is an asymptotic confidence interval for the  measure of distance from normality
under alternative distributions that is inherent in the procedure. Furthermore, there is the option for a neighborhood-of-model validation procedure.
The results of a simulation study, presented in  Section \ref{secsimul}, show
that the novel test is strong with respect to prominent competitors.
In Section \ref{secdataex}, the new tests are applied to the Iris flower data set due to R.A. Fisher.
The paper concludes with some remarks. For the sake of readability, most of the proofs and some
auxiliary results are deferred to Section \ref{secappendix}. Finally, the following abbreviations, valid for $t,x \in \R^d$, will be used in several sections:
\begin{equation}\label{defcsplusminus}
{\rm CS}^{+}(t,x) := \cos(t^\top x) + \sin(t^\top x), \quad {\rm CS}^{-}(t,x) = \cos(t^\top x) - \sin(t^\top x).
\end{equation}

%
%
%

\section{The limits $a \to \infty$ and $a \to 0$}\label{secinfty}
The results of this section show that the class of tests for multivariate normality based on $T_{n,a}$ is
"closed at the boundaries" $a \to \infty$ and $a \to 0$ and thus  shed some light on the tuning parameter $a$, which figures in the weight function $w_a$ given in (\ref{defwa}).
We first consider the case $a \to \infty$.

\begin{theorem}\label{thmainfty} Elementwise on the underlying probability space, we have
\begin{equation}\label{limmori}\lim_{a\rightarrow\infty} \frac{2a^{d/2+1}}{n\pi^{\frac d2}}\, T_{n,a} = \frac1{n^2}\sum_{j,k=1}^n\|Y_{n,j}\|^2\|Y_{n,k}\|^2Y_{n,j}^\top Y_{n,k}.
\end{equation}
\end{theorem}

\begin{proof} For short, put $Y_j := Y_{n,j}$. From the representation of $T_{n,a}$ we have
\begin{eqnarray*}
\frac{2a^{d/2+1}T_{n,a}}{n\pi^{d/2}} & = & \frac{2a}{n^2} \sum_{j,k =1}^n \|Y_j\|^2\|Y_k\|^2\exp\left(-\frac1{4a}\|Y_{n,j}-Y_k\|^2\right)\\
& & - \left(\frac{2a}{2a\! +\! 1}\right)^{d/2+1} \! \!  \frac{2}{n(2a\! +\! 1)}  \sum_{j=1}^n\|Y_{j}\|^2\left(\|Y_{j}\|^2+2da(2a\! +\! 1)\right)\exp\left(-\frac12\frac{\|Y_{j}\|^2}{2a\! +\! 1}\right)\\
& & + \left(\frac{a}{a+1}\right)^{d/2+1} \frac{2}{a+1} \left(a(a+1)d^2+\frac{d(d+2)}{4}\right)\\
& =: & U_{n,1} - U_{n,2} + U_{n,3}
\end{eqnarray*}
(say). Since $\sum_{j=1}^n \|Y_j\|^2 = nd$, an expansion of the exponential function yields
\[
U_{n,1} = 2ad^2 - \frac{d}{n} \sum_{j=1}^n \|Y_j\|^4 + \frac{1}{n^2} \sum_{j,k=1}^n \|Y_j\|^2 \|Y_k\|^2 Y_j^\top Y_k + o(1)
\]
as $a \to \infty$. To tackle $U_{n,2}$, we use
\[
\left(\frac{2a}{2a\! +\! 1}\right)^{d/2+1} = \left( 1+ \frac{1}{2a}\right)^{-(d/2+1)} = 1 - \left(\frac{d}{2}+1\right) \frac{1}{2a} + O(a^{-2})
\]
and, after some algebra, obtain
\[
U_{n,2} = 4ad^2 - \left(\frac{d}{2}+1\right)2d^2 - \frac{d}{n} \sum_{j=1}^n \|Y_j\|^4 + o(1).
\]
Finally, a binomial expansion of $(a/(a+1))^{d/2+1}$ yields
\[
U_{n,3} = 2ad^2 - \left(\frac{d}{2}+1\right) 2d^2 + o(1).
\]
Summing up, the assertion follows.
\end{proof}

\begin{remark}
The limit
\[
\widetilde{b}_{1,d} :=  \frac1{n^2}\sum_{j,k=1}^n\|Y_{n,j}\|^2\|Y_{n,k}\|^2Y_{n,j}^\top Y_{n,k}
\]
(say), which figures on the right-hand side of (\ref{limmori}),
is a measure of multivariate (sample) skewness,  introduced by  {M\'ori}, {Rohatgi} and {Sz\'ekely}, see \cite{MRS:1993}.
A much older time-honored measure of multivariate (sample) skewness is skewness in the sense of Mardia (see \cite{MAR:1970}), which
is given by
\[
b_{1,d} := \frac{1}{n^2} \sum_{j,k=1}^n \left(Y_{n,j}^\top Y_{n,k}\right)^3.
\]
It is interesting to compare Theorem \ref{thmainfty} with similar results found in connection with other weighted $L^2$-statistics
that have been studied for testing $H_0$. Thus, by Theorem 2.1 of \cite{HN:1997}, the time-honored
class of BHEP-statistics for testing for multivariate normality (see \cite{HW:1997}), after suitable rescaling, approaches the linear combination $2b_{1,d} + 3 \widetilde{b}_{1,d}$,
as a smoothing parameter (called $\beta$ in that paper) tends to $0$. Since $\beta$ and $a$ are related by $\beta = a^{-1/2}$, this corresponds to letting $a$ tend to infinity.
The same linear combination $2b_{1,d} + 3 \widetilde{b}_{1,d}$ also showed up as a limit statistic in \cite{HJG:2019} and \cite{HJM:2019}.
Notice that, in the univariate case, the
limit statistic $\widetilde{b}_{1,d}$ figuring in Theorem \ref{thmainfty} is nothing but three times squared sample skewness.
We stress that tests for multivariate normality based on $b_{1,d}$ or $\widetilde{b}_{1,d}$ or on related measures of multivariate skewness and kurtosis
lack consistency against general alternatives, see, e.g., \cite{BH:1991,BH:1992,HEN:1994b,HEN:1994a,H:1997}.
\end{remark}

We now consider the case $a \to 0$.  Since, elementwise on the underlying probability space, the expressions in (\ref{tna2}) and (\ref{tna3})
have finite limits as $a \to 0$, and since the double sum figuring in (\ref{tna1}) converges to $\sum_{j=1}^n \|Y_{n,j}\|^4$ as $a \to 0$,
we have the following result.

\begin{theorem}\label{thmazero} Elementwise on the underlying probability space, we have
\begin{equation}\label{limkurt}\lim_{a\rightarrow 0} \left(\frac{a}{\pi}\right)^{d/2} T_{n,a} =  \frac1n  \sum_{j=1}^n \|Y_{n,j}\|^4.
\end{equation}
\end{theorem}

\begin{remark}
The limit statistic on the right-hand side of (\ref{limkurt}) is Mardia's celebrated measure
$b_{2,d}$ of multivariate sample kurtosis, see \cite{MAR:1970}. Together with Theorem \ref{thmainfty}, this result shows that,
just like the class of BHEP tests for multivariate normality (see \cite{HN:1997}), also the class of tests based on $T_{n,a}$ is "closed at the boundaries"
$a \to \infty $ and $a \to 0$. Notably, Mardia's measure of kurtosis shows up for the first time in connection with limits of
weighted $L^2$-statistics for testing for multivariate normality. The corresponding limit statistic for the class of BHEP tests is, up to a linear transformation,
$n^{-1}\sum_{j=1}^n \exp(-\|Y_{n,j}\|^2/2)$, see Theorem 3.1 of       \cite{HN:1997}.
\end{remark}

%
%
%

\section{A basic Hilbert space central limit theorem}\label{basicCLT}
In this chapter, we present a basic Hilbert space central limit theorem. This theorem implies the limit distribution of  $T_{n,a}$ under the null hypothesis  (\ref{H0}),
but it is also beneficial for proving a limit normal distribution of $T_{n,a}$ under fixed alternatives to $H_0$. Throughout this section,
we assume that the underlying distribution satisfies $\E \|X\|^4 < \infty$. Moreover, we let $\E(X) =0$ and $\E(X X^\top) = {\rm I}_d$ in view of affine invariance.
To motivate the benefit of a Hilbert space setting and for later purposes, it will be convenient to represent $T_{n,a}$ in a different way.

\begin{prop}
Recall $\psi(t)$ from (\ref{defpsit}), and
let
\begin{equation}\label{defbt}
m(t) := \left(d-\|t\|^2\right) \psi(t), \quad t \in \R^d,
\end{equation}
\begin{equation}\label{defznt}
Z_n(t) := \frac{1}{\sqrt{n}} \sum_{j=1}^n \Big{\{} \|Y_{n,j}\|^2 \big{(}\cos (t^\top Y_{n,j}) + \sin(t^\top Y_{n,j})\big{)} - m(t) \Big{\}}, \quad t \in \R^d.
\end{equation}
We then have
\begin{equation}\label{glintqtna}
T_{n,a} = \int Z_n^2(t) \, w_a(t) \, {\rm d} t.
\end{equation}
\end{prop}

\begin{proof}
The proof follows by straightforward algebra using the addition theorems for the sine function and the cosine function and the fact that
$\int \sin(t^\top y) m(t) w_a(t) \, \textrm{d} t=0$, $y \in \R^d$.
\end{proof}

A convenient setting for asymptotics will be the separable Hilbert space $\HH$ of (equivalence classes of) measurable functions $f:\R^d \rightarrow \R$
satisfying $\int f^2(t) w_a(t) \, {\rm d}t < \infty$. The scalar product and the  norm in $\HH$ will be denoted by
\[
\langle f,g \rangle_{\HH} = \int f(t)g(t) \, w_a(t) \, {\rm d}t, \quad \|f\|_\HH = \langle f,f \rangle_\HH^{1/2}, \quad f,g \in \HH,
\]
respectively. Notice that
\[
T_{n,a} = \|Z_n\|_{\HH}^2.
\]
Recalling ${\rm CS}^{\pm}(t,x)$ from (\ref{defcsplusminus}), and putting
\begin{equation}\label{defmuoft}
\mu(t) := \E \big{[} \|X\|^2 {\rm CS}^{+}(t,X) \big{]}, \quad t \in \R^d,
\end{equation}
the main object of this section is the random element $V_n$ of $\HH$, defined by
\begin{equation}\label{defvnt}
V_n(t) := \frac{1}{\sqrt{n}} \sum_{j=1}^n \left( \|Y_{n,j}\|^2 {\rm CS}^{+}(t,Y_{n,j}) - \mu(t) \right), \quad t \in \R^d.
\end{equation}
Observe that $V_n = Z_n$ if the distribution of $X$ is N$_d(0,{\rm I}_d)$, since then the functions $\mu$ and $m$ coincide.
We will show that, as $n \to \infty$, $V_n$ converges in distribution to a centred Gaussian random element $V$ of $\HH$.
The only technical problem in proving such a result is the fact that $V_n$ is based on the scaled residuals $Y_{n,1},\ldots,Y_{n,n}$ and not
on $X_1,\ldots,X_n$. If $V_n^0(t)$ denotes the modification of $V_n(t)$ that results from replacing $Y_{n,j}$ with $X_j$, a
Hilbert space central limit theorem holds for $V_n^0$, since the summands comprising $V_n^0(t)$ are i.i.d. square-integrable centred
random elements of $\HH$. Writing $\vertk$ for convergence in distribution of random elements of $\HH$ and random variables, the basic idea
to prove $V_n \vertk V$ is to find a random element $\widetilde{V}_n$ of $\HH$, such that $\widetilde{V}_n \vertk V$ and
$\widetilde{V}_n - V_n = o_\PP(1)$ as $n \to \infty$. In what follows, the stochastic Landau symbol $o_\PP(1)$ refers to
convergence to zero in probability in $\HH$, i.e., we have to show
\begin{equation}\label{vtildvapprox}
\|\widetilde{V}_n - V_n\|^2_{\HH} = \int \left( \widetilde{V}_n(t) - V_n(t)\right)^2 w_a(t) \, {\rm d}t  = o_\PP(1) \quad \text{as } n \to \infty.
\end{equation}
To state the main result of this section, let $\psi_X(t) = \E [ \exp({\rm i}t^\top X)]$, $t \in \R^d$, denote the characteristic function of $X$, and put
\[
\psi_X^{+}(t) := {\rm Re}\, \psi_X(t)  + {\rm Im} \, \psi_X(t), \quad \psi_X^{-}(t) := {\rm Re} \, \psi_X(t)  - {\rm Im} \, \psi_X(t),
\]
where Re$\, w$ and Im$\, w$ stand for the real and the imaginary part of a complex number $w$, respectively.
For a twice continuously differentiable function $f: \R^d \rightarrow \R$, let H$f(t)$ denote the Hessian matrix
of $f$, evaluated at $t$. Furthermore, recall the gradient operator $\nabla$ and the Laplace operator $\Delta$ from
Section \ref{sec:Intro}.

\begin{prop}\label{propapproxi}
Let
\begin{equation}\label{gldefvntild}
\widetilde{V}_n(t) := \frac{1}{\sqrt{n}} \sum_{j=1}^n v(t,X_j), \quad t \in \R^d,
\end{equation}
where
\begin{eqnarray}\label{reprsumvj}
v(t,x) & = & v_1(t,x) + v_2(t,x) + v_3(t,x) + v_4(t,x),\\ \label{v1v2}
v_1(t,x) & = & \|x\|^2 {\rm CS}^{+}(t,x), \ v_2(t,x) = \frac12 t^\top \left(xx^\top - {\rm I}_d \right) \nabla \Delta \psi_X^{+}(t),\\ \label{v3v4}
v_3(t,x) & = & \left(2 \nabla \psi_X^{-}(t) + \Delta \psi_X^{-}(t)t \right)^\top x, \ v_4(t,x) = x^\top {\rm H} \psi_X^{+}(t)x.
\end{eqnarray}
We then have (\ref{vtildvapprox}).
\end{prop}

The proof of Proposition \ref{propapproxi} is given in Section \ref{secappendix}.\\

Since $\E (X) =0$, $\E (XX^\top) = {\rm I}_d$ and $\E[\|X\|^2 {\rm CS}^{+}(t,X)] = - \Delta \psi_X^{+}(t)$, we have (writing tr for trace)
$\E v(t,X) =  - \Delta \psi_X^{+}(t) + {\rm tr} (H \psi_X^{+}(t))  = 0$.
Thus, $v(\cdot,X_1), \ldots, v(\cdot,X_n)$ are i.i.d. centred square-integrable random elements of $\HH$, and the central limit theorem in
Hilbert spaces gives $\widetilde{V}_n \vertk V$ for some centred  Gaussian element $V$ of $\HH$.  In view of (\ref{vtildvapprox}) and Slutsky's lemma,
we therefore can state the main result of this section.

\begin{theorem}\label{themmainclt}
Let $X,X_1,X_2, \ldots $ be i.i.d. random vectors satisfying $\E \|X\|^4 < \infty$, $\E (X) =0$ and $\E(XX^\top ) = {\rm I}_d$.
For the sequence of random elements $V_n$ defined in (\ref{defvnt}) we have
\[
V_n \vertk V \quad \text{as } n \to \infty,
\]
where $V$ is a centred Gaussian element of $\HH$ having  covariance kernel
\begin{equation}\label{defkernl}
L(s,t) = \E [ v(s,X)v(t,X)], \quad s,t \in \R^d,
\end{equation}
where $v(t,x)$ is given in (\ref{reprsumvj}).
\end{theorem}

%
%
%

\section{The limit null distribution of $T_{n,a}$}\label{secnull}

In this section we derive the limit distribution of $T_{n,a}$ under the null hypothesis  (\ref{H0}). In view of affine invariance, we assume
that $X$ has a $d$-variate standard normal distribution. Since the process $Z_n(t)$ given in (\ref{defznt}) is nothing  but $V_n$, as  defined in (\ref{defvnt}),
in this special case, we have the following result.

\begin{theorem}\label{thmsec3main} Suppose that $X$ has some non-degenerate normal distribution. Putting
\[
d_2 := d+2, \quad d_4 := d+4,
\]
we have the following:
\begin{enumerate}
\item[a)] There is a centred Gaussian random element $Z$ of $\HH$ with covariance kernel
\begin{eqnarray*}
K(s,t)&=&\psi(s\! -\! t)\Big{(}\left(\|s\! -\! t\|^2\! -\! d_2\right)^2 \! -\!  2d_2\Big{)} +\psi(s)\psi(t)\Big{\{}-\frac{(s^\top t)^2}{2}(\|s\|^2\! -\! d_4)(\|t\|^2\! -\! d_4)\\
&& +2d_2\left(\|s\|^2\! +\! \|t\|^2\right) \! -\!  \|s\|^4 \! -\!  \|t\|^4\! -\! \|s\|^2\|t\|^2 \! - \! s^\top t \left(\|s\|^2\! -\! d_2\right)\left(\|t\|^2\! -\! d_2\right)\! -\! dd_2\Big{\}},
\end{eqnarray*}
$s,t \in \R^d$, such that, with $Z_n$ defined in (\ref{defznt}), we have
$
Z_n \vertk Z \textrm{ in } \HH \textrm{ as } n \to \infty$.
\item[b)]
We have
\[
T_{n,a} \vertk \int Z^2(t) \, w_a(t) \, {\rm d}t.
\]
\end{enumerate}
\end{theorem}

Notice that b) follows from a) and the continuous mapping theorem.
Part a) follows from Theorem \ref{themmainclt}. For the special case $X \sim {\rm N}_d(0,{\rm I}_d)$, we have
$\psi_X^{+}(t) = \psi_X^{-}(t) = \exp(-\|t\|^2/2) = \psi(t)$, which entails $\nabla \psi(t)  =  - t \psi(t)$,
$\Delta \psi(t)  =  \left(\|t\|^2-d \right) \psi(t)$,
${\rm H} \psi(t)  =  \left( tt^\top - {\rm I}_d \right) \psi(t)$,  and
$\nabla \Delta \psi(t)  =  t \psi(t) \left( 2+d-\|t\|^2 \right)$.
Thus, the function $v(t,x)$ figuring in the statement of Proposition \ref{propapproxi} takes the special form
\begin{eqnarray}\label{defhtx}
h(x,t) & = & \|x\|^2 {\rm CS}^{+}(t,x) - (2\psi(t) + m(t)) t^\top x - \psi(t)\|x\|^2 \\ \nonumber
& & + \left(2 \psi(t) + \frac{m(t)}{2}\right) (t^\top x)^2 - \left( \psi(t) + \frac{m(t)}{2}\right) \|t\|^2.
\end{eqnarray}
Long but straightforward computations, using symmetry arguments and the identities
\begin{eqnarray*}
\E\left[\|X\|^2(s^\top X)^2\right]&=&(d+2)\|s\|^2,\\
\E\left[\|X\|^4\cos(t^\top X)\right]&=&\left(\left(d+2-\|t\|^2\right)\left(d-\|t\|^2\right)-2\|t\|^2\right)\psi(t),\\
\E\left[\|X\|^2s^\top X \sin(t^\top X)\right]&=&s^\top t\left(d+2-\|t\|^2\right)\psi(t),\\
\E\left[\|X\|^2(s^\top X)^2 \cos(t^\top X)\right]&=&\left(d+2-\|t\|^2\right)\left(\|s\|^2-(s^\top t)^2\right)\psi(t)-2(s^\top t)^2\psi(t),
\end{eqnarray*}
$s,t\in\R^d$, show that the covariance kernel $K(s,t) = \E[h(s,X)h(t,X)]$ takes the form given above.

\vspace*{3mm}

Let $T_{\infty,a}$ be a random variable with the limit null distribution of $T_{n,a}$, i.e., with the distribution of
$\int Z^2(t) w_a(t) \, {\rm d}t$. Since
$
\E (T_{\infty,a}) = \int K(t,t) w_a(t) \, {\rm d}t$,
the following result may be obtained by straightforward but tedious manipulations of integrals.

\begin{theorem}\label{thmsec3ewert}
Putting
\[
c_j(a,d) := \frac{\pi^{d/2} d}{(a+1)^{d/2 +j}}, \quad j =1,2,3,4,
\]
we have
\begin{eqnarray*}
\E(T_{\infty,a}) & = & d(d+2)\left(\left(\frac{\pi}{a}\right)^{d/2} -  \left(\frac{\pi}{a+1}\right)^{d/2} \right)\\
& & - c_4(a,d) \frac{(d+2)(d+4)(d+6)}{32}+c_3(a,d)\frac{(d+2)(d+3)(d+4)}{8}\\
& & - c_2(a,d) \frac{(d+2)(d^2+4d+14)}{8} + c_1(a,d) \frac{(d-2)(d+2)}{2}.
\end{eqnarray*}
\end{theorem}

The quantiles of the distribution of $T_{\infty,a}$ can be approximated by a Monte Carlo method, see Section \ref{secsimul}.
%
%
%

\section{Contiguous alternatives}\label{seccontig}
In this section, we consider a triangular array $(X_{n1}, \ldots, X_{nn})$, $n \ge d+1$, of
rowwise i.i.d. random vectors with Lebesgue density
\[
f_n(x) = \varphi(x) \left( 1 + \frac{g(x)}{\sqrt{n}}\right), \quad x \in \R^d,
\]
where $\varphi(x) = (2\pi)^{-d/2} \exp \left(-\|x\|^2/2\right)$, $x \in \R^d$, is the density of the
standard normal distribution N$_d(0,{\rm I}_d)$, and $g$ is some bounded measurable function  satisfying
$
\int g(x) \varphi(x) \, {\rm d} x = 0$.
We assume that $n$ is sufficiently large to render $g$ nonnegative. Recall $Z_n(t)$ from (\ref{defznt}).

\begin{theorem}\label{thmcontiguous} Under the triangular $X_{n,1}, \ldots, X_{n,n}$ given above, we have
\[
Z_n \vertk Z + c \quad \text{ as } n \to \infty,
\]
where $Z$ is the centred random element of $\HH$ figuring in Theorem \ref{thmsec3main}, and
\[
c(t) = \int h(x,t) g(x) \varphi(x) \, {\rm d}x,
\]
with $h(x,t)$ given in (\ref{defhtx}).
\end{theorem}

\begin{proof}
Since the reasoning uses standard LeCam theory on contiguous probability measures and closely parallels  that given in Section 3 of \cite{HW:1997},
it will be omitted.
\end{proof}

\begin{corollary}\label{corocontig} Under the conditions of Theorem \ref{thmcontiguous}, we have
\[
T_{n,a} \vertk \int \left(Z(t) + c(t)\right)^2 w_a(t) \, {\rm d} t.
\]
\end{corollary}

From Theorem \ref{thmcontiguous} and the above corollary, we conclude that the test for
multivariate normality based on $T_{n,a}$ is able to detect alternatives which converge to the normal distribution at
the rate $n^{-1/2}$, irrespective of the underlying dimension $d$. The test of Bowman and Foster (see \cite{BF:1993}) is
a prominent example of an affine invariant tests for multivariate normality
that is consistent against each fixed non-normal alternative distribution but nevertheless lacks this property of $n^{-1/2}$-consistency,
see Section 7 of \cite{H:2002}.

%
%
%
\section{Fixed alternatives and consistency}\label{secfixed}
In this section we assume that $X,X_1,X_2, \ldots $ are i.i.d. with a distribution that is absolutely continuous with respect to Lebesgue measure, and that
$\E \|X\|^4 < \infty$. In view of affine invariance, we make the additional assumptions $\E(X) =0$ and $\E(XX^\top) = {\rm I}_d$. Recall $m(t)$ from
(\ref{defbt}) and $\mu(t)$ from (\ref{defmuoft}).  Writing $\fsk$ for $\PP$-almost sure convergence, our first result is a strong limit for $T_{n,a}/n$.

\begin{theorem}\label{thmfixed1}
If $\E \|X\|^4 < \infty$, we have
\begin{equation}\label{defdeltaa}
\frac{T_{n,a}}{n} \fsk \Delta_a= \int \big{(} \mu(t) - m(t)\big{)}^2 w_a(t) \, {\rm d}t \quad  \text{ as } n \to \infty.
\end{equation}
\end{theorem}

\noindent The proof of Theorem \ref{thmfixed1} is given in Section \ref{secappendix}.

\begin{remark}
Let $\psi_X(t) = \E [\exp({\rm i}t^\top X)]$, $t \in \R^d$, denote the characteristic function of $X$. We have
$\Delta \psi_X(t) = - \E[\|X\|^2 \exp({\rm i}t^\top X)]$. Since $\Delta \psi(t) = (\|t\|^2-d)\psi(t)$, some algebra gives
\[
\Delta_a = \int \Big{|} \Delta \psi_X(t) - \Delta \psi(t) \Big{|}^2 w_a(t) \, {\rm d} t.
\]
By Theorem \ref{unique.sol}, we have $\Delta_a=0$ if and only if $X$ has the normal distribution N$_d(0,{\rm I}_d)$.
In view of Theorem  \ref{thmfixed1}, $T_{n,a} \to \infty $ $\PP$-almost  surely under any alternative distribution
satisfying $\E \|X\|^4 < \infty$. Since, according to Theorem \ref{thmsec3main},  the sequence of critical values of a level-$\alpha$-test of $H_0$ that rejects $H_0$ for
large values of $T_{n,a}$ stays bounded, we have the following result.
\end{remark}

\begin{corollary} The test for multivariate normality based on $T_{n,a}$ is consistent against any fixed alternative
distribution satisfying $\mathbb{E} \|X\|^4 < \infty$.
\end{corollary}

By analogy with  Theorem \ref{thmainfty}, the next result shows that  the (population) measure of multivariate skewness in the sense  of {M\'ori}, {Rohatgi} and {Sz\'ekely} (see \cite{MRS:1993})
emerges as the limit of $\Delta_a$, after a suitable normalization, as $a \to \infty$.

\begin{theorem}\label{thmlimdeltaa}  If $\E \|X\|^6 < \infty$, then
\[
\lim_{a \to \infty} \frac{2 a^{d/2+1}}{\pi^{d/2}}  \Delta_a = \big{\|} \E (\|X\|^2 X)\big{\|}^2.
\]
\end{theorem}

\noindent The proof of Theorem \ref{thmlimdeltaa} is given in Section \ref{secappendix}.

\vspace*{3mm}
We now show that the limit distribution of $\sqrt{n}\left(T_{n,a}/n- \Delta_a\right)$ as $n \to \infty $ is centred normal.
This fact is essentially a consequence of Theorem 1 in \cite{BEH:2017}. The reasoning is as follows:
By (\ref{glintqtna}), we have $T_{n,a} = \|Z_n\|^2_{\HH}$,
where $Z_n$ is given in (\ref{defznt}). Putting $z(t) := \mu(t)-m(t)$, $t \in \R^d$,  display (\ref{defdeltaa}) shows that
$\Delta_a = \|z\|^2_{\HH}$.
Now, defining $Z_n^*(t) := n^{-1/2} Z_n(t)$, it follows that
\begin{eqnarray} \nonumber
\sqrt{n}\left( \frac{T_{n,a}}{n} - \Delta_a \right) & = & \sqrt{n} \left( \|Z_n^*\|^2_{\HH} - \|z\|^2_{\HH} \right) = \sqrt{n} \langle Z_n^* -z,Z_n^*+z \rangle_\HH \\ \nonumber
& = &  \sqrt{n} \langle Z_n^* -z,2z + Z_n^*-z \rangle_{\HH}\\ \label{zerlegung1}
& = & 2 \langle \sqrt{n}(Z_n^* -z),z \rangle_{\HH} + \frac{1}{\sqrt{n}} \| \sqrt{n}\left(Z_n^* - z\right)\|^2_{\HH}.
\end{eqnarray}
A little thought shows that $\sqrt{n}(Z_n^*(t) - z(t)) = V_n(t)$, where $V_n$ is given in (\ref{defvnt}).
By Theorem \ref{themmainclt}, $V_n \vertk V$ for a centred Gaussian element $V$ of $\HH$.
 As a consequence, the  second summand in (\ref{zerlegung1}) is $o_\PP(1)$ as $n \to \infty$, and, by the continuous mapping theorem,
 the first summand
converges in distribution to $2\langle V,z\rangle_\HH$. The latter random variable has a centred normal distribution with variance $\sigma_a^2 := 4 \E [\langle V,z \rangle_\HH ^2]$.
The following theorem summarizes our findings.

\begin{theorem}\label{themfixed2}
For a fixed alternative distribution satisfying $\E \|X\|^4 < \infty$, $\E(X) =0$ and $\E(XX^\top) = {\rm I}_d$, we have
\[
\sqrt{n}\left( \frac{T_{n,a}}{n} - \Delta_a \right) \vertk {\rm N}(0,\sigma_a^2) \qquad \text{as } n \to \infty,
\]
where
\begin{equation}\label{darstsigma2}
\sigma_a^2 = 4 \iint L(s,t) z(s) z(t) \, w_a(s) w_a(t) \, {\rm d}s{\rm d}t
\end{equation}
and $L(s,t)$ is given in (\ref{defkernl}).
\end{theorem}

\begin{proof}
To complete the proof, notice that, by Fubini's theorem,
\begin{eqnarray*}
\sigma_a^2  &  = &  4 \E [\langle V,z \rangle_\HH ^2] \ = \
 4 \E \Bigg{[} \left(\int V(s) \, z(s) \, w_a(s) \, {\rm d} s\right)  \left(\int V(t) \, z(t) \, w_a(t) \, {\rm d} t\right) \Bigg{]}\\
& = &  4 \iint \E[V(s)V(t)] \, z(s) \, z(t) \, w_a(s) w_a(t) \, {\rm d}s{\rm d}t. \qedhere
\end{eqnarray*}
\end{proof}

Under slightly stronger conditions on $\mathbb{P}^X$, there is a consistent estimator of $\sigma_a^2$.
To obtain such an estimator, we replace $L(s,t)$  as well as $z(s)$ and $z(t)$ figuring in (\ref{darstsigma2})
with suitable empirical counterparts. To this end, notice that, by (\ref{defkernl}), we have
\begin{equation}\label{deflst}
L(s,t) = \sum_{i=1}^4 \sum_{j=1}^4 L^{i,j}(s,t),
\end{equation}
where
\begin{equation}\label{deflijst}
L^{i,j}(s,t) = \E\big{[} v_i(s,X)v_j(t,X)\big{]}
\end{equation}
and $v_j(t,x)$, $j \in \{1,2,3,4\}$, are given in (\ref{v1v2}), (\ref{v3v4}).
  Since $\nabla \Delta \psi_X^{\pm}(t) =  \mp \E[{\rm CS}^{\mp}(t,X) \|X\|^2X]$, $\nabla \psi_X^{\pm}(t) = \pm \E[{\rm CS}^{\mp}(t,X)X]$,
$\Delta \psi_X^{\pm}(t) = - \E[{\rm CS}^{\pm}(t,X)\|X\|^2]$ and ${\rm H}\psi_X^{+}(t) = - \E[{\rm CS}^{+}(t)(t,X)$ $X X^\top]$, parts a) -- d) of the following lemma show that the unknown quantities $\nabla \Delta \psi_X^{+}(t)$, $\nabla \psi_X^-(t)$,
$\Delta\psi_X^{-}(t)$ and $\text{H}\psi_X^+(t)$ that figure in the expressions of $v_2(t,x)$, $v_3(t,x)$ and $v_4(t,x)$ can be replaced with consistent
estimators that are based on the scaled residuals $Y_{n,1}, \ldots,Y_{n,n}$ defined in (\ref{defscaledr}).

\begin{lemma} \label{lemma6.6}
If $\E\|X_1\|^6 < \infty$, we have
\begin{enumerate}
\vspace{-1.5mm}
\item[\emph{a)}] $\Psi_{1,n}(t) := n^{-1}\sum_{j=1}^n {\rm CS}^+(t, Y_{n,j})Y_{n,j} \fse -\nabla\psi_X^-(t),$
\item[\emph{b)}] $\Psi_{2,n}(t) := n^{-1}\sum_{j=1}^n {\rm CS}^+(t, Y_{n,j})Y_{n,j}Y_{n,j}^{\top} \fse -\emph{H}\psi_X^+(t),$
\item[\emph{c)}] $\Psi_{3,n}^{\pm}(t) := n^{-1}\sum_{j=1}^n {\rm CS}^{\pm}(t, Y_{n,j})\|Y_{n,j}\|^2 \fse -\Delta\psi_X^{\pm}(t),$
\item[\emph{d)}] $\Psi_{4,n}^{\pm}(t) := n^{-1}\sum_{j=1}^n {\rm CS}^{\pm}(t, Y_{n,j})\|Y_{n,j}\|^2Y_{n,j} \fse \pm\nabla\Delta\psi_X^{\mp}(t),$
\item[\emph{e)}] $\Psi_{5,n}(t) := n^{-1}\sum_{j=1}^n {\rm CS}^+(t, Y_{n,j})\|Y_{n,j}\|^2Y_{n,j}Y_{n,j}^{\top} \fse \emph{H}\Delta\psi_X^+(t)$.
\end{enumerate}
\end{lemma}

\noindent The proof of Lemma~\ref{lemma6.6} is given in Section 9.

\vspace*{2mm}
\noindent
In view of Lemma \ref{lemma6.6}, a suitable estimator of $L(s,t)$ defined in (\ref{deflst}) is
\begin{equation}\label{deflnst}
L_n(s,t) = \sum_{i=1}^4 \sum_{j=1}^4 L_n^{i,j}(s,t),
\end{equation}
where
\begin{equation}\label{deflnij}
L_n^{i,j}(s,t) = \frac{1}{n} \sum_{k=1}^n v_{n,i}(s,Y_{n,k})v_{n,j}(t,Y_{n,k}),
\end{equation}
and
\begin{eqnarray*}
v_{n,1}(s,Y_{n,k}) & = & \|Y_{n,k}\|^2 {\rm CS}^{+}(s,Y_{n,k}), \quad v_{n,2}(s,Y_{n,k}) = - \frac{1}{2} s^\top (Y_{n,k}Y_{n,k}^\top - {\rm I}_d) \Psi_{4,n}^{-}(s),\\
v_{n,3}(s,Y_{n,k}) & = & - \left(2\Psi_{1,n}(s) + \Psi_{3,n}^{-}(s)s\right)^\top Y_{n,k}, \quad v_{n,4}(s,Y_{n,k}) = -Y_{n,k}^\top \Psi_{2,n}(s) Y_{n,k}.
\end{eqnarray*}

\noindent By straightforward algebra we have
\begin{eqnarray}\nonumber
    L^{1,1}_n(s, t) &=& \textstyle{n^{-1}\sum_{j=1}^n \|Y_{n,j}\|^4\cos\bigl((t-s)^{\top}Y_{n,j}\bigr) + n^{-1}\sum_{j=1}^n \|Y_{n,j}\|^4\sin\bigl((t+s)^{\top}Y_{n,j}\bigr),} \\ \label{ln12}
    L^{1,2}_n(s, t) &=& \textstyle{-\frac{1}{2}\Psi_{4,n}^-(s)^{\top}\Psi_{5,n}(t)s + \frac{1}{2}\Psi_{3,n}^+(t)\Psi_{4,n}^-(s)^{\top}s,} \\ \nonumber
    L^{1,3}_n(s, t) &=& \textstyle{\bigl(-2\Psi_{1,n}(s) - \Psi_{3,n}^-(s)s\bigr)^{\top}\Psi_{4,n}^+(t)} \\ \nonumber
    L^{1,4}_n(s, t) &=& \textstyle{-n^{-1}\sum_{j=1}^n \|Y_{n,j}\|^2{\rm CS}^+(t, Y_{n,j})Y_{n,j}^{\top}\Psi_{2,n}(s)Y_{n,j}} \\ \nonumber
    L^{2,2}_n(s, t) &=& \textstyle{\frac{1}{4}\Psi_{4,n}^-(t)^{\top}n^{-1}\sum_{j=1}^n Y_{n,j}Y_{n,j}^{\top}t\Psi_{4,n}^-(s)^{\top}\left(Y_{n,j}Y_{n,j}^{\top} - \text{I}_d\right)s,} \\ \nonumber
    L^{2,3}_n(s, t) &=& \textstyle{\Psi_{4,n}^-(t)^{\top}n^{-1}\sum_{j=1}^n Y_{n,j}Y_{n,j}^{\top}t\left(\Psi_{1,n}(s) + \frac{1}{2}\Psi_{3,n}^-(s)s\right)^{\top}Y_{n,j},} \\ \nonumber
    L^{2,4}_n(s, t) &=& \textstyle{\frac{1}{2}\Psi_{4,n}^-(t)^{\top}n^{-1}\sum_{j=1}^n \left(Y_{n,j}Y_{n,j}^{\top} - \text{I}_d\right)tY_{n,j}^{\top}\Psi_{2,n}(s)Y_{n,j},} \\ \nonumber
    L^{3,3}_n(s, t) &=& \textstyle{\bigl(2\Psi_{1,n}(t) + \Psi_{3,n}^-(t)t\bigr)^{\top}\bigl(2\Psi_{1,n}(s) + \Psi_{3,n}^-(s)s\bigr),} \\ \nonumber
    L^{3,4}_n(s, t) &=& \textstyle{n^{-1}\sum_{j=1}^n \left(2\Psi_{1,n}(s) + \Psi_{3,n}^-(s)s\right)^{\top}Y_{n,j}Y_{n,j}^{\top}\Psi_{2,n}(t)Y_{n,j},} \\ \nonumber
    L^{4,4}_n(s, t) &=& \textstyle{n^{-1}\sum_{j=1}^n Y_{n,j}^{\top}\Psi_{2,n}(t)Y_{n,j}Y_{n,j}^{\top}\Psi_{2,n}(s)Y_{n,j}.}
\end{eqnarray}
Notice that, by symmetry, $L_n^{i,j} = L_n^{j,i}$ if $i \neq j$.
Since $z(s) = \mu(s)-m(s) = \E[\|X\|^2 {\rm CS}^{+}(s,X)] - m(s)$, a natural estimator of $z(s)$ is
\begin{equation}\label{defznsx}
z_n(s) = \frac{1}{n}\sum_{k=1}^n {\rm CS}^{+}(s,Y_{n,k}) \|Y_{n,k}\|^2 -m(s).
\end{equation}
Writing $\stackrel{\mathbb{P}}{\longrightarrow}$ for convergence in probability, we then have the following result.

\begin{theorem}\label{thmvarest}
Suppose $\mathbb{E} \|X\|^6 < \infty$, $\E(X) =0$ and $\E(XX^\top) = {\rm I}_d$. Let
\begin{equation}\label{defhatsigma2}
\widehat{\sigma}_{n,a}^2 = 4 \iint L_n(s,t) \, z_n(s) \, z_n(t) \, w_a(s) w_a(t) \, {\rm d}s{\rm d}t,
\end{equation}
where $L_n(s,t)$ and $z_n(s)$ are defined in (\ref{deflnst}) and (\ref{defznsx}), respectively. We then have
\[
\widehat{\sigma}_{n,a}^2 \stackrel{\mathbb{P}}{\longrightarrow} \sigma_a^2.
\]
\end{theorem}

The proof of Theorem \ref{thmvarest} is given in Section \ref{secappendix}.\\

Under the conditions of Theorem \ref{thmvarest}, Theorem \ref{themfixed2} and Sluzki's lemma yield
\begin{equation}\label{neighborh}
\frac{\sqrt{n}}{\widehat{\sigma}_{n,a}}\left( \frac{T_{n,a}}{n} - \Delta_a \right) \vertk {\rm N}(0,1) \qquad \text{as } n \to \infty,
\end{equation}
provided that $\sigma_a^2 >0$. We thus obtain the following asymptotic confidence interval for $\Delta_a$.

\begin{corollary}
Let $\alpha \in (0, 1)$, and write $\Phi_{1-\alpha/2}$ for the $(1-\alpha/2)$-quantile of the standard normal law. Then
\[I_{n,a,\alpha} := \left[\frac{T_{n,a}}{n} - \frac{\widehat{\sigma}_{n,a}}{\sqrt{n}}\Phi_{1-\alpha/2}, \frac{T_{n,a}}{n} + \frac{\widehat{\sigma}_{n,a}}{\sqrt{n}}\Phi_{1-\alpha/2}\right]\]
is an asymptotic confidence interval for $\Delta_a$.
\end{corollary}

\begin{example}
We consider the case $d=1$, $a=0.1$ and the following standardized symmetric alternative distributions:
        the uniform distribution U$\left(-\sqrt{3}, \sqrt{3}\right)$, the Laplace distribution
         L$\left(0, 1/\sqrt{2}\right)$, and the logistic distribution Lo$\left(0, \sqrt{3}/\pi\right)$.
         The corresponding characteristic functions and their second derivatives are given by
\begin{align*}
    \varphi_{{\rm U}}(t) &= \frac{\sin(\sqrt{3}t)}{\sqrt{3}t}, & \varphi_{{\rm U}}''(t) &= \frac{\sqrt{3}(2 - 3t^2)\sin(\sqrt{3}t) - 6t\cos(\sqrt{3}t)}{3t^3}, \\
    \varphi_{{\rm L}}(t) &= \frac{2}{2 + t^2}, & \varphi_{{\rm L}}''(t) &= \frac{12t^2 - 8}{(2 + t^2)^3}, \\
    \varphi_{{\rm Lo}}(t) &= \frac{\sqrt{3}t}{\sinh(\sqrt{3}t)}, & \varphi_{{\rm Lo}}''(t) &= \frac{3}{2}\frac{3\sqrt{3}t - 2\sinh(2\sqrt{3}t) + \sqrt{3}t\cosh(2\sqrt{3}t)}{\sinh(\sqrt{3}t)^3}.
\end{align*}
The pertaining values of $\Delta_{0.1}$ are $\Delta_{0.1,{\rm U}} \approx \text{0.3322}$, $\Delta_{0.1,{\rm L}} \approx \text{0.127}$ and $\Delta_{0.1,{\rm Lo}} \approx \text{0.033}$.
Table \ref{Tab:covering} shows the percentages of coverage of the confidence intervals $I_{n,0.1,\alpha}$
for $\alpha = \text{0.05}$ and several sample sizes. Each entry is based on 5000 Monte-Carlo-replications.
The results are quite satisfactory for $n \ge 50$.

\begin{table}[t]
    \centering
    \renewcommand{\arraystretch}{1.1}
    \begin{tabular}{c|c|c|c}

    $n$ & U$\left(-\sqrt{3}, \sqrt{3}\right)$ & L$\left(0, 1/\sqrt{2}\right)$ & Lo$\left(0, \sqrt{3}/\pi\right)$ \\
    \hline
    $20$ & 91.5 & 87.5 & 83.4 \\
    $50$ & 93.6 & 96.9 & 95.2 \\
    $100$ & 94.4 & 97.8 & 98.4 \\
    $200$ & 94.8 & 97.9 & 99.0 \\
    $300$ & 94.5 & 97.5 & 98.9 \\

    \end{tabular}
    \caption{Percentages of coverage of $\Delta_{0.1}$ by $I_{n,0.1,\alpha}$ for different distributions (5000 replications)}\label{Tab:covering}
\end{table}
\end{example}

\begin{remark} A further application of (\ref{neighborh}) addresses a fundamental drawback inherent in any goodness of fit test, see \cite{BEH:2017}.
If a level-$\alpha$-test of $H_0$ does not reject $H_0$, the hypothesis $H_0$ is by no means {\em validated} or {\em confirmed}, since each model is wrong, and
there is probably only not enough evidence to reject $H_0$. However, if we define a certain "essential distance"
 $\delta_0 >0$, we can consider the inverse testing problem
\[
H_{\delta_0}: \Delta_a(F) \ge \delta_0 \  \textrm{ against } \ K_{\delta_0}: \Delta_a(F) < \delta_0.
\]
Here, the dependence of $\Delta_a$ on the underlying distribution has been made explicit.

From \eqref{neighborh}, we obtain an asymptotic level-$\alpha$-{\em neighborhood-of-model-validation}
 test of $H_{\delta_0}$ against $K_{\delta_0}$, which rejects $H_{\delta_0}$ if and only if
\[
\frac{T_{n.a}}{n} \ \le \ \delta_0 - \frac{\widehat{\sigma}_{n,a}}{\sqrt{n}} \Phi^{-1}(1-\alpha).
\]
Indeed, from  \eqref{neighborh} we have for each $F \in H_{\delta_0}$
\begin{eqnarray*}
& &
\limsup_{n\to \infty} \PP_F \! \left(\frac{T_{n.a}}{n}   \le  \delta_0 - \frac{\widehat{\sigma}_{n,a}}{\sqrt{n}} \Phi^{-1}(1\! -\! \alpha)\! \right)\\
  &  =  &  \limsup_{n\to \infty} \PP_F \! \left(\frac{\sqrt{n}}{\widehat{\sigma}_{n,a}}\! \left(\frac{T_{n,a}}{n} \!  - \! \delta_0\! \right) \le -\Phi^{-1}(1\! -\! \alpha) \! \right)\\
& \le &
 \alpha .
\end{eqnarray*}
Moreover, it follows that
\[
\lim_{n\to \infty}\PP_F \left(\frac{T_{n,a}}{n}  \le  \delta_0 - \frac{\widehat{\sigma}_{n,a}}{\sqrt{n}} \Phi^{-1}(1-\alpha)\right) \ = \ \alpha
\]
for each $F$ such that $\Delta_a(F) = \delta_0$. Finally, we have
\[
\lim_{n\to \infty} \PP_F\left(\frac{T_{n,a}}{n}  \le  \delta_0 - \frac{\widehat{\sigma}_{n,a}}{\sqrt{n}} \Phi^{-1}(1-\alpha) \right) = 1
\]
if $\Delta_a(F) < \delta_0$. Thus, the test is consistent against each alternative.
\end{remark}

\begin{remark}
The double integral figuring (\ref{defhatsigma2}) may be calculated explicitly. To this end, notice that
$\widehat{\sigma}_{n,a}^2 = \sum_{i,j=1}^4 \widehat{\sigma}_{n,a}^{i,j}$,
  where
\begin{equation}\label{defsigmanij}
\widehat{\sigma}_{n,a}^{i,j} = 4 \iint L_n^{i,j}(s,t) z_n(s)z_n(t) w_a(s)w_a(t) \, {\rm d}s{\rm d}t
\end{equation}
and $\widehat{\sigma}_{n,a}^{i,j} = \widehat{\sigma}_{n,a}^{j,i}$, $i,j \in \{1,\ldots,4\}$, by symmetry.
We put
\begin{align*}
    q_{1,a}(y) &:= \int m(t) {\rm CS}^+(t, y)w_a(t){\rm d}t = \frac{(2\pi)^{d/2}}{(2a+1)^{d/2+2}}\bigl(\|y\|^2 + 2da(2a+1)\bigr)\exp\left(-\frac{1}{2}\frac{\|y\|^2}{2a+1}\right), \\
    p_{1,a}(y,z) &:= \int {\rm CS}^+(t, y){\rm CS}^+(t, z)w_a(t){\rm d}t = \left(\frac{\pi}{a}\right)^{d/2}\exp\left(-\frac{\|y - z\|^2}{4a}\right), \\
    p_{2,a}(y,z) &:= \int {\rm CS}^+(t, y) {\rm CS}^-(t,z)tw_a(t){\rm d}t = \left(\frac{\pi}{a}\right)^{d/2}\frac{1}{2a}\exp\left(-\frac{\|y - z\|^2}{4a}\right)(y-z), \\
    q_{2,a}(y) &:= \displaystyle\int m(t){\rm CS}^-(t, y)tw_a(t){\rm d}t \\
    &= \frac{(2\pi)^{d/2}}{(2a+1)^{d/2+3}}\left(2(2a+1)(1-ad) - \|y\|^2\right)\exp\left(-\frac{1}{2}\frac{\|y\|^2}{2a+1}\right)y
\end{align*}
for $y,z \in \mathbb{R}^d$ and
\begin{eqnarray*}
    P_n^{1,a,1} & := & \frac{1}{n^2} \sum_{j,k=1}^n \|Y_j\|^2 \|Y_k\|^2 p_{1,a}(Y_j, Y_k) - \frac{1}{n}\sum_{j=1}^n \|Y_j\|^2 q_{1,a}(Y_j), \\
    \widetilde{P}_n^{1,a,1} & := & \frac{1}{n^2}\sum_{j,k=1}^n \|Y_j\|^2\|Y_k\|^2p_{1,a}(Y_j,Y_k)Y_k - \frac{1}{n}\sum_{j=1}^n \|Y_j\|^2q_{1,a}(Y_j)Y_j, \\
    \widetilde{P}_n^{1,a,2} & := & \frac{1}{n^2}\sum_{j,k=1}^n \|Y_j\|^2p_{1,a}(Y_j, Y_k)Y_k - \frac{1}{n}\sum_{j=1}^n q_{1,a}(Y_j)Y_j, \\
    \overline{P}_n^{1,a} & := & \frac{1}{n^2}\sum_{j,k=1}^n \|Y_j\|^2\|Y_k\|^2p_{1,a}(Y_j, Y_k)Y_kY_k^{\top} - \frac{1}{n} \sum_{j=1}^n \|Y_j\|^2Y_jY_j^{\top}q_{1,a}(Y_j), \\
    P_n^{1,a,2}(Y_j) & := & \frac{1}{n^2}\sum_{k,\ell =1}^n   \|Y_k\|^2\left(Y_j^{\top}Y_\ell \right)^2p_{1,a}(Y_k, Y_l) - \frac{1}{n}\sum_{k=1}^n \left(Y_j^{\top}Y_k\right)^2q_{1,a}(Y_k), \\
    P_n^{1,a,3}(Y_j) & := & \frac{1}{n}\sum_{k=1}^n \|Y_k\|^2p_{1,a}(Y_j, Y_k) - q_{1,a}(Y_j), \\
    P_n^{2,a,1} & := & \frac{1}{n^2}\sum_{j,k=1}^n \|Y_j\|^2\|Y_k\|^2Y_k^{\top}p_{2,a}(Y_j, Y_k) - \frac{1}{n}\sum_{j=1}^n \|Y_j\|^2Y_j^{\top}q_{2,a}(Y_j), \\
    P_n^{2,a,2} & := & \frac{1}{n^2}\sum_{j,k=1}^n \|Y_j\|^2\|Y_k\|^2Y_k^{\top}\overline{P}_{1,a} p_{2,a}(Y_j, Y_k) - \frac{1}{n}\sum_{j=1}^n \|Y_j\|^2Y_j^{\top}\overline{P}_{1,a}q_{2,a}(Y_j), \\
    \widetilde{P}_n^{2,a} & := & \frac{1}{n^2}\sum_{j,k=1}^n \|Y_j\|^2\|Y_k\|^2p_{2,a}(Y_j, Y_k) - \frac{1}{n}\sum_{j=1}^n \|Y_j\|^2q_{2,a}(Y_j), \\
    P_n^{2,a,3}(Y_j) & := & \frac{1}{n^2}\sum_{k, \ell =1}^n  \|Y_k\|^2\|Y_\ell\|^2Y_k^{\top}Y_jY_j^{\top}p_{2,a}(Y_k, Y_\ell) - \frac{1}{n}\sum_{k=1}^n \|Y_k\|^2Y_k^{\top}Y_jY_j^{\top}q_{2,a}(Y_k),
\end{eqnarray*}
where $Y_j$ is shorthand for $Y_{n,j}$.
Notice that $\widetilde{P}_n^{1,a,1}$, $\widetilde{P}_n^{1,a,2}$ and $\widetilde{P}_n^{2,a}$ are vectors and
 $\overline{P}_n^{1,a}$ is a matrix. By straightforward but tedious manipulations of the integrals in (\ref{defsigmanij}), each $\widehat{\sigma}_n^{i,j}$
 is seen to be an  arithmetic mean of functions of the scaled residuals, namely:
\begin{eqnarray*}
    \widehat{\sigma}^{1,1}_{n,a} & = & \frac{4}{n}\sum_{j=1}^n \|Y_j\|^4P_n^{1,a,3}(Y_j)^2, \quad
    \widehat{\sigma}^{1,2}_{n,a} = 2P_n^{1,a,1}P_n^{2,a,1} - 2P_n^{2,a,2}, \\
    \widehat{\sigma}^{1,3}_{n,a} & = & -4\left(2\widetilde{P}_n^{1,a,2} + \widetilde{P}_n^{2,a}\right)^{\top}\widetilde{P}_n^{1,a,1}, \quad
    \widehat{\sigma}^{1,4}_{n,a} = -\frac{4}{n}\sum_{j=1}^n \|Y_j\|^2P_n^{1,a,2}(Y_j)P_n^{1,a,3}(Y_j),    \\
        \widehat{\sigma}^{2,2}_{n,a} & = & \frac{1}{n}\sum_{j=1}^n P_{2,a,3}(Y_j)^2 - (P_n^{2,a,1})^2, \quad
    \widehat{\sigma}^{2,3}_{n,a} = \frac{4}{n}\sum_{j=1}^n P_{2,a,3}(Y_j)\left(\widetilde{P}_n^{1,a,2} + \frac{1}{2}\widetilde{P}_n^{2,a}\right)^{\top}Y_j,  
      \end{eqnarray*}
    \begin{eqnarray*}
    \widehat{\sigma}^{2,4}_{n,a} & = & \frac{2}{n}\sum_{j=1}^n \bigl(P_n^{2,a,3}(Y_j) - P_n^{2,a,1}\bigr)P_n^{1,a,2}(Y_j), \quad
    \widehat{\sigma}^{3,3}_{n,a} = 4 \cdot \left\|2\widetilde{P}_n^{1,a,2} + \widetilde{P}_n^{2,a}\right\|^2, \\
    \widehat{\sigma}^{3,4}_{n,a} & = & \frac{4}{n}\sum_{j=1}^n P_n^{1,a,2}(Y_j)\left(2\widetilde{P}_n^{1,a,2} + \widetilde{P}_n^{2,a}\right)^{\top}Y_j, \quad
    \widehat{\sigma}^{4,4}_{n,a} = \frac{4}{n}\sum_{j=1}^n P_n^{1,a,2}(Y_j)^2.
\end{eqnarray*}
\end{remark}

%
%
%

\section{Simulations}\label{secsimul}
In this section, we present the results of a Monte Carlo simulation study on the finite-sample power performance of the test based on $T_{n,a}$ with that of
 several competitors. Since different procedures have been used for univariate and multivariate data, we distinguish the cases $d=1$ and $d \ge 2$.
 In the univariate case, the sample sizes are $n\in\{20,50,100\}$, and in the multivariate case
  we restrict the simulations to $n\in\{20,50\}$. The nominal level of significance is fixed throughout all simulations to $0.05$.
  Critical values for $T_{n,a}$  (in fact, for a scaled version of $T_{n,a}$ in order to obtain values of a similar magnitude across the range of values for $d$ and $a$ considered)
  have been simulated under $H_0$ with 100~000 replications, see Table \ref{Tab:emp.crit}. The critical values in the rows in Table \ref{Tab:emp.crit} denoted by "$\infty$" represent approximations of quantiles of the distribution of the limit random element $T_{\infty,a}=\int Z^2(t) \, w_a(t) \, {\rm d}t$ in Theorem \ref{thmsec3main} b). To obtain these values we simulate (say) $m$ i.i.d. random supporting points $U_1,\ldots,U_m$ where $U_1\sim\mbox{N}_d(0,(2a)^{-1}{\rm I}_d)$, which are adapted to the integration over $\R^d$ with respect to the weight function $w_a(t)=\exp(-a\|t\|^2)$  and a large number (say) $\ell$ of random variables $Z_j=\frac1{d^2m}\|X_j\|^2$, $j=1,\ldots,\ell$, with i.i.d. $X_j\sim \mbox{N}_m(0,\Sigma_K)$, where the covariance matrix $\Sigma_k$ is given by $\Sigma_K=\left(K(U_{k_1},U_{k_2})\right)_{k_1,k_2\in\{1,\ldots,m\}}$ and $K$ is the covariance kernel in Theorem \ref{thmsec3main} a). Next, we calculate the empirical $95\%$ quantile of $Z_1,\ldots,Z_\ell$. Each approximation was simulated with $\ell=100,000$ and $m=1,000$ for $d\in \{2,3,5,10\}$ and each entry in Tables \ref{pow.T.1} - \ref{pow.T.5}, which exhibit percentages of rejection of $H_0$ of the tests under consideration against various alternative distributions,  is based on 10~000 replications. Entries that are marked with $\ast$ denote $100\%$.
  The simulations have been performed using the statistical computing environment {\tt R}, see \cite{R:2019}.

\begin{table}[t]
\centering
\begin{tabular}{rr|rrrrrr}
$d$ & $n\backslash a$ & 0.25 & 0.5 & 1 & 1.5 & 2 & 3 \\
  \hline
  \multirow{4}{*}{1}& 20 & 2.730 & 2.085 & 1.597 & 1.420 & 1.320 & 1.180 \\
   & 50 & 2.950 & 2.299 & 1.761 & 1.566 & 1.455 & 1.301 \\
  & 100 & 2.954 & 2.323 & 1.798 & 1.583 & 1.474  & 1.322 \\
  & $\infty$ & 2.839 & 2.348 & 1.765 & 1.581 & 1.495  & 1.289\\
  \hline
  \multirow{4}{*}{2}& 20 & 1.916 & 1.482 & 0.957 & 0.726 & 0.622 & 0.529 \\
   & 50 & 1.978 & 1.551 & 1.039 & 0.800 & 0.689 & 0.592 \\
  & 100 & 1.963 & 1.543 & 1.044 & 0.815 & 0.707 & 0.608 \\
  & $\infty$ & 2.018 & 1.489 & 1.013 & 0.794 & 0.674  & 0.598\\
   \hline
   \multirow{4}{*}{3}& 20 & 1.610 & 1.304 & 0.841 & 0.591 & 0.461 & 0.351 \\
   & 50 &1.662 & 1.357 & 0.904 & 0.654 & 0.521 & 0.405 \\
  &100 & 1.649 & 1.348 & 0.899 & 0.656 & 0.526 & 0.414 \\
  & $\infty$ & 1.617 & 1.314 &  0.857 & 0.635 & 0.505  & 0.408\\
  \hline
  \multirow{4}{*}{5} & 20 & 1.360 & 1.201 & 0.835 & 0.575 & 0.417 & 0.263 \\
   & 50 &  1.421 & 1.261 & 0.903 & 0.645 & 0.482 & 0.315 \\
  & 100 &  1.421 & 1.259 & 0.901 & 0.647 & 0.486 & 0.323 \\
  & $\infty$ & 1.437 &  1.260 &  0.862 & 0.659 & 0.485  & 0.312\\
  \hline
  \multirow{4}{*}{10}& 20 & 1.129 & 1.109 & 0.960 & 0.750 & 0.569 & 0.337 \\
  & 50 &  1.205 & 1.181 & 1.033 & 0.832 & 0.652 & 0.409 \\
  & 100 & 1.219 & 1.194 & 1.043 & 0.842 & 0.664 & 0.422 \\
  & $\infty$ & 1.284 &  1.258 &  1.061 & 0.867 & 0.696  & 0.427
\end{tabular}
\caption{Empirical and approximate quantiles for $d^{-2}\left(\frac{a}{\pi}\right)^{d/2}T_{n,a}$ and $\alpha=0.05$ (100~000 replications)}\label{Tab:emp.crit}
\end{table}

\subsection{Testing univariate normality}
For testing the hypothesis $H_0$ that the distribution of $X$ is univariate normal (i.e., $\mathbb{P}^X \in \mathcal{N}_1$),
we considered the following competitors to the new test statistic:
\begin{enumerate}
\item the Shapiro--Wilk test (SW), see \cite{SW:1965},
\item the Shapiro--Francia test (SF), see \cite{SF:1972},
\item the Anderson--Darling test (AD), see \cite{AD:1952},
\item the Baringhaus--Henze--Epps--Pulley test (BHEP), see \cite{HW:1997},
\item the del Barrio--Cuesta-Albertos--M\'{a}tran--Rodr\'{i}guez-Rodr\'{i}guez test (BCMR), see \cite{DBCMRR:1999},
\item the Betsch--Ebner test (BE), see \cite{BE:2019}.
\end{enumerate}
For the implementation of the first three tests in {\tt R}
we refer to the package {\tt nortest} by \cite{GL:2015}. Tests based on the empirical characteristic function
 are represented by the BHEP-test with tuning parameter $a > 0$. The statistic is given in (\ref{BHEP.test}),
 $a = 1$ is fixed, and the critical values are taken from \cite{Hen:1990}.

The BCMR-test is based on the $L^2$-Wasserstein distance, see section 3.3 in \cite{delBarrio2000}, and is given by
\begin{equation*}
	{\rm BCMR} =
	n \left( 1 - \frac{1}{S_n^2} \left( \sum_{k = 1}^n X_{(k)} \int_{\frac{k - 1}{n}}^{\frac{k}{n}} \Phi^{-1}(t) \, \mathrm{d}t \right)^2 \right) - \int_{\frac{1}{n + 1}}^{\frac {n}{n + 1}} \frac{t (1 - t)}{\left( \varphi\left( \Phi^{-1}(t) \right) \right)^2} \, \mathrm{d}t.
\end{equation*}
Here, $X_{(k)}$ is the $k$-th order statistic of $X_1, \dots, X_n$, $S_n^2$ is the sample variance,
 and $\Phi^{-1}$ is the quantile function of the standard normal law. Simulated critical values can be found in \cite{K:2009}, or in Table 4 of \cite{BE:2019}.

The BE-test is based on a $L^2$-distance between the empirical zero-bias transformation and the empirical distribution.
Since the zero-bias transformation has a unique fixed point under normality, this distance is minimal under $H_0$. The statistic is given by
\begin{eqnarray*} \label{explicit formula stat1}
	{\rm BE}_a & = &  \frac{2}{n} \sum\limits_{1 \leq j < k \leq n} \left\{ \vphantom{\exp\left(- \tfrac{Y_{(k)}^2}{2 a}\right)} \left( 1 - \Phi\left( \tfrac{Y_{(k)}}{\sqrt{a}} \right) \right) \left( (Y_{(j)}^2 - 1)(Y_{(k)}^2 - 1) + a Y_{(j)} Y_{(k)} \right) \right.\\
	 & &\left. + \frac{a}{\sqrt{2 \pi a}} \, \exp\left(- \tfrac{Y_{(k)}^2}{2 a}\right) \left( - Y_{(j)}^2 Y_{(k)} + Y_{(k)} + Y_{(j)} \right)
	\right\} \\
	& & + \frac{1}{n} \sum\limits_{j=1}^{n} \left\{ \vphantom{\frac{a}{\sqrt{2 \pi a}}} \left( 1 - \Phi\left(\tfrac{Y_{j}}{\sqrt{a}}\right) \right) \left( Y_j^4 + (a - 2) Y_j^2 + 1 \right) \right. \\
	& &\left. + \frac{a}{\sqrt{2 \pi a}} \, \exp\left(- \tfrac{Y_j^2}{2 a}\right) \left( 2 Y_j - Y_j^3 \right) \right\},
\end{eqnarray*}
where $Y_{1}, \dots, Y_{n}$ is shorthand for the scaled residuals $Y_{n, 1}, \dots, Y_{n,n}$, $Y_{(1)} \leq \dotso \leq Y_{(n)}$ are the order statistics
of $Y_1,\ldots,Y_n$, and $\Phi$ stands for the distribution function of the standard normal law.
  The implementation employs a bootstrap procedure to find a data driven version of the tuning parameter $a$, see \cite{AS:2015}.
  We used $B=400$ bootstrap replications and the same grid of tuning parameters as in \cite{BE:2019}, p. 19.

The alternative distributions are chosen to fit the extensive power study of univariate normality tests by \cite{RDC:2010},
 in order to ease the comparison with a wide variety of other existing tests. As representatives of symmetric distributions
 we simulate the Student t$_\nu$-distribution with $\nu \in \{3, 5, 10\}$ degrees of freedom, as well as the uniform distribution U$(-\sqrt{3}, \sqrt{3})$.
 The asymmetric distributions are represented by the $\chi^2_\nu$-distribution with $\nu \in \{5, 15\}$ degrees of freedom,
 the Beta distributions B$(1, 4)$ and B$(2, 5)$, the Gamma distributions $\Gamma(1, 5)$ and $\Gamma(5, 1)$, parametrized by their shape and rate parameter,
 the Gumbel distribution Gum$(1, 2)$ with location parameter 1 and scale parameter 2, the lognormal distribution LN$(0, 1)$
 as well as the Weibull distribution W$(1, 0.5)$ with scale parameter 1 and shape parameter 0.5.
 As representatives of bimodal distributions we take the mixture of normal distributions NMix$(p, \mu, \sigma^2)$, where the random variables are generated by
\begin{equation*}
	(1 - p) \, {\rm N}(0, 1) + p \, {\rm N}(\mu, \sigma^2), \quad p \in (0, 1), \, \mu \in \R, \, \sigma > 0.
\end{equation*}

Table \ref{pow.T.1} shows that the empirical power estimates of the new test $T_{n,a}$ outperform the other strong
procedures for the symmetric $t$-distribution,  and they can compete for most of the other alternatives.
Interestingly, the power does not differ too much when varying the tuning parameter $a$,
although an effect is clearly visible, especially for the uniform distribution.
A data driven choice as in \cite{AS:2015} might be of benefit also in connection with the new testing procedure.

\subsection{Testing multivariate normality}
In this subsection we consider testing the hypothesis $H_0$ that the distribution of $X$ is multivariate normal (i.e., belongs to $\mathcal{N}_d$),
for the dimensions $d \in \{2,3,5\}$. As competitors to the new test statistic we chose
\begin{enumerate}
\item the Henze--Visagie test (HV), see \cite{HV:2019},
\item the Henze--Jim\'{e}nez-Gamero test (HJG), see \cite{HJG:2019},
\item the Baringhaus--Henze--Epps--Pulley test (BHEP), see \cite{HW:1997}.
\end{enumerate}
The HV-test uses a weighted $L^2$-type statistic based on a characterization of the moment generating function
that employs a system of first-order partial differential equations. The statistic is defined by
\begin{equation*}
	{\rm HV}_\gamma= \frac{1}{n}\left(\frac{\pi}{\gamma}\right)^{\frac{d}{2}} \sum_{j,k = 1}^n \exp\left( \frac{\|Y_{n,j} + Y_{n,k}\|^2}{4 \gamma} \right)\left(Y_{n,j}^\top Y_{n,k}+\|Y_{n,j} + Y_{n,k}\|^2\left(\frac{1}{4 \gamma^2}-\frac{1}{2 \gamma}\right)+\frac{d}{2\gamma}\right),
\end{equation*}
where $\gamma>2$. We followed the recommendation of the authors in \cite{HV:2019} and fixed $\gamma=5$.
Since the limiting statistic HV$_\infty$ for $\gamma\rightarrow\infty$ is a linear combination of sample skewness
in the sense of Mardia and that of {M\'ori}, {Rohatgi} and {Sz\'ekely}, we also included HV$_\infty$.

The HJG-test uses a weighted $L^2$-distance between the empirical moment generating function of the standardized sample
 and the moment generating function of the standard normal distribution. The test statistic is given by
\begin{equation*}
	{\rm HJG}_\beta= \frac{1}{n\beta^{\frac{d}{2}}} \sum_{j,k = 1}^n \exp\left( \frac{\|Y_{n,j} + Y_{n,k}\|^2}{4 \beta} \right) - \frac{2}{\sqrt{\beta - 1/2}}\sum_{j = 1}^n \exp\left( \frac{\|Y_{n,j}\|^2}{4 \beta - 2}\right)+ \frac{n}{(\beta - 1)^\frac{d}{2}},
\end{equation*}
with $\beta > 0$. In our simulation we fix $\beta=1.5$. For each of the tests based on HV$_{5}$, HV$_\infty$ and HJG$_{1.5}$,
 critical values were simulated with 100~000 replications.

Finally, the now classical BHEP-test examines the weighted $L^2$-distance between the empirical characteristic function
 of the standardized data and the characteristic function of the $d$-variate standard normal distribution. The statistic has the simple form
\begin{eqnarray}\label{BHEP.test}
	{\rm BHEP}_a&=& \frac{1}{n^2} \sum_{j,k = 1}^n \exp\left( -\frac{a^2}2\|Y_{n,j} - Y_{n,k}\|^2 \right)\\&& - 2\left(1+a^2\right)^{-\frac{d}{2}}\frac1n\sum_{j = 1}^n \exp\left( -\frac{a^2\|Y_{n,j}\|^2}{2\left(1+a^2\right)}\right)+ (1+2a^2))^{-\frac{d}{2}},\nonumber
\end{eqnarray}
with a tuning parameter $a>0$. A variety of values of $a$, i.e., $a\in\{0.1,0.25,0.5,0.75,1,2,3,5,10\}$, has been considered.
 Critical values can be found in Tables I - III of \cite{HW:1997}, whereas missing critical values have been simulated separately with 100~000 replications.

The alternative distributions are chosen to fit the simulation study in \cite{HV:2019} and are defined as follows. We denote by NMix$(p,\mu,\Sigma)$ the normal mixture distributions generated by
\begin{equation*}
	(1 - p) \, {\rm N}_d(0, {\rm I}_d) + p \, {\rm N}_d(\mu, \Sigma), \quad p \in (0, 1), \, \mu \in \R^d, \, \Sigma > 0,
\end{equation*}
where $\Sigma > 0$ stands for a positive definite matrix. We write in the notation of above $\mu=3$ for a $d$-variate
vector of 3's and $\Sigma={\rm B}_d$ for a $(d \times d)$-matrix containing 1's on the main diagonal and 0.9's on each off-diagonal entry.
 We denote by $t_\nu(0,{{\rm I}}_d)$ the multivariate $t$-distribution with $\nu$ degrees of freedom, see \cite{AG:2009}.
 By DIST$^d(\vartheta)$ we denote the $d$-variate random vector generated by independently simulated components of the distribution
 DIST with parameter vector $(\vartheta)$, where DIST is taken to be the Cauchy distribution C, the logistic distribution L,
  the Gamma distribution $\Gamma$ as well as the Pearson Type VII  distribution P$_{VII}$, with $\vartheta$ denoting in this specific
   case the degrees of freedom. The spherical symmetric distributions where simulated using the {\tt R} package {\tt distrEllipse},
  see \cite{RKSC:2006}, and are denoted by $\mathcal{S}^d(\mbox{DIST})$, where DIST stands for the distribution of the radii
  and was chosen to be the exponential, the beta and the $\chi^2$-distribution.

From Tables \ref{pow.T.2} to \ref{pow.T.5}, it is obvious that $T_{n,a}$ outperforms the competing tests for most
of the alternatives considered, again showing that the tuning parameter has --compared to the BHEP test-- little effect on the power.
 As in the univariate case $T_{n,a}$ has very strong power against the multivariate $t$-distribution.
 If the radial distribution of the spherical symmetric alternatives has compact support, the BHEP test exhibits
  a better performance than $T_{n,a}$. The HV$_5$- and the HJG$_{1.5}$-test have a good power,  but they are
   mostly dominated by the BHEP-test and the $T_{n,a}$-test. Again, a data driven choice of the tuning parameter $a$
   would be beneficial for the test, but to the best of our knowledge no reliable multivariate method is yet available.
   We suggest to choose a small tuning parameter like $a=0.25$.

%
%
%

\section{Real Data Example: The Iris Data set}\label{secdataex}
In 1936 R.A. Fisher presented the now classical data set called {\tt Iris Flower}, see Table I in \cite{F:1936}.
The data consist of the four variables sepal, length, sepal width, petal length, and petal width, measured
on fifty specimens of each of three types or iris, namely {\it Iris setosa}, {\it Iris versicolor}, and {\it Iris virginica}.
This data set is included in the statistical language {\tt R}, and it can be downloaded from the UCI Machine Learning Repository,
see \cite{DG:2019}. That reference provides a list of articles that use this specific data set to validate clustering methods,
neural networks or learning algorithms, and it presents a typical test case for statistical classification techniques in machine learning,
such as support vector machines. A visualization of two-dimensional projections of the data set is given in Figure \ref{fig:iris}.

\begin{figure}[t]
\centering
\includegraphics[scale=0.4]{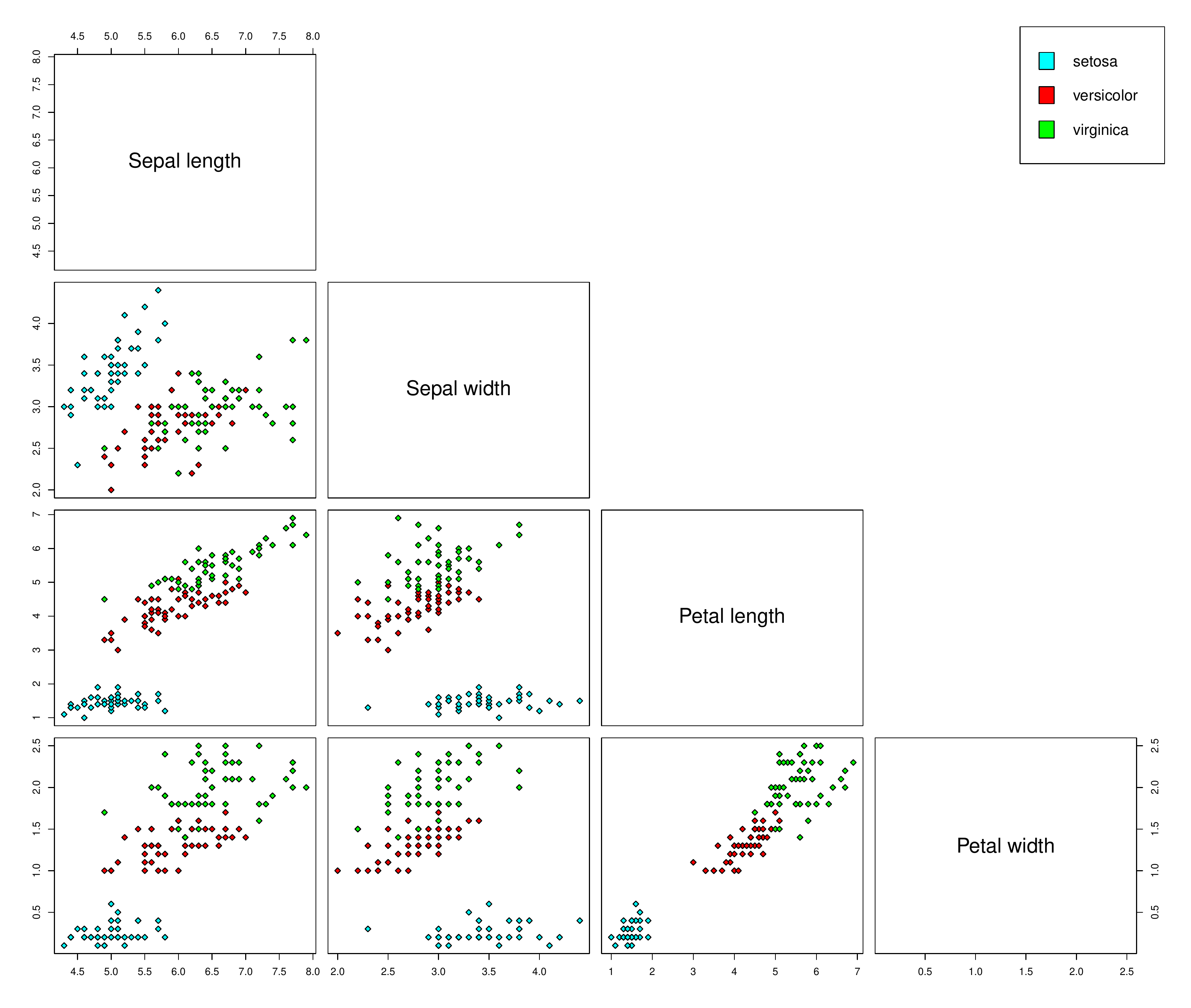}
\caption{2D projections of the {\tt iris} data set with coloured species.}\label{fig:iris}
\end{figure}

\begin{table}[b]
\centering
\begin{tabular}{lr|ccccccc}
Species& $ a$ & 0.25 & 0.5 & 1 & 2 & 3 & 5 & 10\\ \hline
setosa &     & 0.0631 & 0.0706 & 0.0683 & 0.0431 & 0.0386 & 0.0555 & 0.0918 \\
versicolor &  & 0.4402 & 0.3560 & 0.2912 & 0.2766 & 0.2707 & 0.2626 & 0.2573 \\
virginica &  & 0.1943 & 0.1671 & 0.1336 & 0.1385 & 0.1643 & 0.2042 & 0.2071\\ \hline
mixed populations & & 0.0000 & 0.0000 & 0.0000 &0.0000 &0.0012 &0.0048 & 0.0150
\end{tabular}
\caption{Empirical $p$-values for $T_{n,a}$.}\label{Tab:real_data}
\end{table}
In Table \ref{Tab:real_data} we present empirical $p$-values simulated with 10~000 replications.
As can be seen, the test does not reject the hypothesis of normality on a small significance
level (like $\alpha=0.01$) for the different species for each of the tuning parameters considered.
For the {\it Iris setosa} data, however, an increase of the significance level to 0.05
results in a rejection of the hypothesis for $a=2$ and $a=3$. For the whole data set,
 we observe a small $p$-value due to the mixture of the three populations and consequently reject the hypothesis $H_0$.

\section{Concluding remarks}\label{secconclude}
We proved consistency of the test for multivariate normality based on $T_{n,a}$ against each alternative distribution that
 satisfies the moment condition $\E \|X\|^4 < \infty$. Intuitively, the  test should be "all the more consistent" if $\E \|X\|^4 = \infty$.
 In fact, we conjecture consistency of the new test against any non-normal alternative distribution.\\
The limiting random element $T_{\infty,a}=\|Z\|^2_\mathbb{H}$ from Theorem \ref{thmsec3main} b)
has the same distribution as $\sum_{j=1}^\infty\lambda_jN_j^2$, where the $N_j$ are i.i.d. standard normal random variables,
 and the $\lambda_j$ are the positive eigenvalues corresponding to eigenfunctions of the linear integral operator
 $Kf(s)=\int_{\R^d}K(s,t)f(t)\, w_a(t) \, {\rm d}t$ associated with the covariance kernel $K$ from Theorem \ref{thmsec3main} a), i.e.,
  we have $\lambda f(s)=\int_{\R^d}K(s,t)f(t)\, w_a(t) \, {\rm d}t.$ These positive eigenvalues clearly depend on $K$
  and the weight function $w_a$. It is hardly possible to obtain a closed form expression for $\lambda_1,\lambda_2,\ldots$. It would be interesting to approximate the eigenvalues either numerically or by some Monte Carlo method, since the largest eigenvalue has a crucial influence on the approximate Bahadur efficiency, see \cite{B:1960} and the monograph \cite{N:1995}, as well as \cite{HNE:2009} for results on distribution-free $L_p$-type statistics.

\section*{Acknowledgement}
The authors thank Ioannis Anapolitanos for his help in proving Theorem \ref{unique.sol} and for sharing his knowledge of Schr\"{o}dinger operators.

%
%
%

\begin{appendix}
\section{Proofs and Auxiliary results}\label{secappendix}
This section contains the proofs of some of the theorems as well as some auxiliary results that are used in the main text.
Recall our standing assumption that the distribution of $X$ is absolutely continuous. In what follows, let
\[
\Delta_{n,j} = (S_n^{-1/2}-\textrm{I}_d)X_j - S_n^{-1/2}\overline{X}_n, \quad j=1,\ldots,n.
\]
Notice that $\Delta_{n,j} = Y_{n,j}-X_j$, where the scaled residuals $Y_{n,j}$ are given in (\ref{defscaledr}).

\begin{prop}\label{propossumselta}
Let $X, X_1,X_2, \ldots $ be i.i.d. random vectors satisfying $\E\|X\|^4 < \infty$,  $\E(X) =0$ and $\E(X X^\top) = {\rm I}_d$.
We then have:
\begin{enumerate}
\item[a)]$\displaystyle{
\sum_{j=1}^n \|\Delta_{n,j}\|^2 = O_{\PP}(1)}.
$
\item[b)] $\displaystyle{\frac{1}{n} \sum_{j=1}^n \|\Delta_{n,j}\|^2 \fsk 0}$.
\item[c)] $\displaystyle{\max_{j=1,\ldots,n} \|\Delta_{n,j}\| = o_\PP\left(n^{-1/4}\right)}$.
\end{enumerate}
\end{prop}

\begin{proof} a) Notice that
\[
\|\Delta_{n,j}\|^2 = X_j^\top(S_n^{-1/2}-\textrm{I}_d)^2X_j - 2 \overline{X}_n^\top S_n^{-1/2} (S_n^{-1/2}-\textrm{I}_d)X_j + \overline{X}_nS_n^{-1}\overline{X}_n
\]
and thus, using tr$(AB) = {\rm tr}(BA)$, where tr$(C)$ denotes the trace of a square matrix $C$,
\begin{equation}\label{prop7sumdelta}
\sum_{j=1}^n \|\Delta_{n,j}\|^2  =  \textrm{tr}\Big{(}(S_n^{-1/2}\! -\!  \textrm{I}_d)^2 \sum_{j=1}^n X_jX_j^\top \Big{)} - 2n\overline{X}_n^\top S_n^{-1/2}(S_n^{-1/2}\! -\!  \textrm{I}_d)\overline{X}_n
+ n\overline{X}_n S_n^{-1} \overline{X}_n .
\end{equation}
Due to $\E\|X_1\|^4 < \infty$, the central limit theorem  implies
\[
\sqrt{n}\left(S_n^{-1} - {\rm I}_d\right) = - S_n^{-1} \sqrt{n}\left( S_n - {\rm I}_d\right) = O_\PP(1).
\]
Since $\sqrt{n}(S_n^{-1/2}- {\rm I}_d)(S_n^{-1/2}+ {\rm I}_d) = \sqrt{n}(S_n^{-1} - {\rm I}_d)$, it follows that $\sqrt{n}(S_n^{-1/2}- {\rm I}_d) = O_\PP(1)$.
In view of $\sum_{j=1}^n X_jX_j^\top = O_\PP(n)$ and $\sqrt{n}\overline{X}_n = O_\PP(1)$, we are done.

b) After dividing both sides of (\ref{prop7sumdelta}) by $n$, the first summand on the right hand side converges to 0 $\PP$-almost surely because
$n^{-1}\sum_{j=1}^n X_jX_j^\top \fsk {\rm I}_d$,  and $S_n^{-1/2} \fsk {\rm I}_d$. The same holds for the other two terms,
since $\overline{X}_n \fsk 0$.

c) Let $\|A\|_{{\rm sp}}$ be the spectral norm of a matrix $A$. Then
\begin{equation}\label{abschdeltx}
\|\Delta_{n,j}\| \le \|S_n^{-1/2} - {\rm I}_d\|_{{\rm sp}} \, \|X_j\| + \|S_n^{-1/2}\|_{{\rm sp}} \|\overline{X}_n\|
\end{equation}
and thus
\[
\max_{j=1,\ldots,n}  \|\Delta_{n,j}\| \le \|S_n^{-1/2} - {\rm I}_d\|_{{\rm sp}} \, \max_{j=1,\ldots,n} \|X_j\| + \|S_n^{-1/2}\|_{{\rm sp}} \|\overline{X}_n\|.
\]
From Theorem 5.2 of \cite{BN:1963} we have $\max_{j=1,\ldots,n}\|X_j\|/n^{1/4} \fsk 0$.
Since $\sqrt{n} \|S_n^{-1/2} - {\rm I}_d\|_{{\rm sp}} = O_\PP(1)$, $\|S_n^{-1/2}\|_{{\rm sp}} = O_\PP(1)$ and $\sqrt{n} \|\overline{X}_n\| = O_\PP(1)$, we are done.
 \end{proof}

\begin{prop}\label{propossumseltaxx}
Let $X, X_1,X_2, \ldots $ be i.i.d. random vectors satisfying $\E\|X\|^6 < \infty$,  $\E(X) =0$ and $\E(X X^\top) = {\rm I}_d$.
We then have
\[
\frac{1}{n} \sum_{j=1}^n \|\Delta_{n,j}\|^k \|X_j\|^\ell  \fsk 0
\]
as $n \to \infty$ for each $k \in \mathbb{N}$ and $\ell \in \mathbb{N}_0$ such that $k+ \ell \le 6$.
\end{prop}

\begin{proof}
From (\ref{abschdeltx}) we have
\[
\|\Delta_{n,j}\|^k \, \|X_j\|^\ell \le 2^k \|S_n^{-1/2} - {\rm I}_d\|_{{\rm sp}}^k \, \|X_j\|^{k+\ell} + 2^k \|S_n^{-1/2}\|_{{\rm sp}}^k \|\overline{X}_n\|^k \|X_j\|^\ell.
\]
Since $\|S_n^{-1/2} - {\rm I}_d\|_{{\rm sp}} \fsk 0$, $S_n^{-1/2} \fsk {\rm I}_d$ and  $\|\overline{X}_n\| \fsk 0$, the assertion follows from
the strong law of large numbers.
\end{proof}

\noindent {\bf Proof of Proposition \ref{propapproxi}}\\[1mm]
In what follows, we put $Y_j = Y_{n,j}$ and $\Delta_j = \Delta_{n,j}$ for the sake of brevity.
Notice that
$\cos(t^\top Y_j) = \cos(t^\top X_j) - \sin(\Theta_j) t^\top \Delta_j$, $\sin(t^\top Y_j) = \sin(t^\top X_j) + \cos(\Gamma_j)t^\top \Delta_j$,
where $\Theta_j, \Gamma_j$ depend on $X_1,\ldots,X_n$ and $t$ and satisfy
\begin{equation}\label{sincosungl}
|\Theta_j - t^\top X_j| \le |t^\top \Delta_j|,  \quad  |\Gamma_j - t^\top X_j| \le |t^\top \Delta_j|.
\end{equation}
Since $\|Y_j\|^2 = \|X_j\|^2 + \|\Delta_j\|^2 + 2 X_j^\top \Delta_j$, it follows that
$V_n(t) = \sum_{k=1}^6 V_{n,k}(t)$,  where
\begin{eqnarray}\label{vn1final}
V_{n,1}(t) & = & \frac{1}{\sqrt{n}} \sum_{j=1}^n \Big{\{} \|X_j\|^2 {\rm CS}^{+}(t,X_j) - \mu(t) \Big{\}},\\ \nonumber
V_{n,2}(t) & = & \frac{1}{\sqrt{n}} \sum_{j=1}^n \|X_j\|^2 t^\top \Delta_j \big{(}\cos(\Gamma_j) - \sin(\Theta_j)\big{)},\\ \nonumber
V_{n,3}(t) & = & \frac{2}{\sqrt{n}} \sum_{j=1}^n X_j^\top \Delta_j {\rm CS}^{+}(t,X_j), \\ \nonumber
V_{n,4}(t) & = & \frac{2}{\sqrt{n}} \sum_{j=1}^n X_j^\top \Delta_j t^\top \Delta_j \big{(}\cos(\Gamma_j) - \sin(\Theta_j)\big{)},\\ \nonumber
V_{n,5}(t) & = & \frac{1}{\sqrt{n}} \sum_{j=1}^n \|\Delta_j\|^2 {\rm CS}^{+}(t,X_j),\\ \nonumber
V_{n,6}(t) & = & \frac{1}{\sqrt{n}} \sum_{j=1}^n \|\Delta_j\|^2 t^\top \Delta_j \big{(}\cos(\Gamma_j) - \sin(\Theta_j)\big{)}.
\end{eqnarray}
We first show that $V_{n,\ell} = o_\PP(1)$ if $\ell \in \{4,5,6\}$. As for $V_{n,4}$, notice that, by the Cauchy--Schwarz inequality,
$|V_{n,4}(t)| \le 4\|t\| n^{-1/2} \max_{i=1,\ldots,n} \|X_i\| \sum_{j=1}^n \|\Delta_j\|^2$.
Since $\E \|X\|^4 < \infty$, Theorem 5.2 of \cite{BN:1963} yields $\max_{i=1,\ldots,n} \|X_i\| = o_\PP(n^{1/4})$. In view of
Proposition \ref{propossumselta} a), we have $V_{n,4} = o_\PP(1)$.  The same proposition immediately also gives
$V_{n,5} = o_\PP(1)$. Since
$
|V_{n,6}(t)| \le 2\|t\| n^{-1/2} \max_{i=1,\ldots,n}\|\Delta_i\| \sum_{j=1}^n \|\Delta_j\|^2$,
we have $V_{n,6} = o_\PP(1)$ in view of Proposition \ref{prop7sumdelta} a) and Proposition \ref{propossumselta} c).

We now consider $V_{n,2}(t)$. Since $|V_{n,2}(t)| \le 2\|t\| n^{-1} \sum_{j=1}^n \|X_j\|^2 n^{1/4}\max_{i=1,\ldots,n}\|\Delta_i\| n^{^1/4}$,
the law of large numbers and Proposition  \ref{prop7sumdelta} c) show that $V_{n,2} = o_\PP(n^{1/4})$. In view of
 (\ref{sincosungl}) and Proposition \ref{prop7sumdelta} c), the error is thus $o_\PP(1)$ if we replace both  $\Gamma_j$ and $\Theta_j$
 with $t^\top X_j$. Moreover, plugging $\Delta_j = (S_n^{-1/2}-{\rm I}_d)X_j - S_n^{-1/2}\overline{X}_n$ into the definition of $V_{n,2}(t)$,
 the error is $o_\PP(1)$ if we replace $S_n^{-1/2}\overline{X}_n$ with $\overline{X}_n$. Recalling (\ref{defcsplusminus}), we thus obtain
 \begin{equation}\label{vn2app1}
 V_{n,2}(t) = \frac{1}{\sqrt{n}} \sum_{j=1}^n \|X_j\|^2 t^\top \Big{\{}\left(S_n^{-1/2} - {\rm I}_d\right) X_j - \overline{X}_n \Big{\}} {\rm CS}^{-}(t,X_j) + o_\PP(1).
 \end{equation}
We now use display (2.13) of \cite{HW:1997}, according to which
\[
2 \sqrt{n}(S_n^{-1/2}- {\rm I}_d) = - \frac{1}{\sqrt{n}} \sum_{j=1}^n \left(X_jX_j^\top - {\rm I}_d \right) + O_\PP\left(n^{-1/2}\right).
\]
Plugging this expression into (\ref{vn2app1}) we obtain
\begin{eqnarray}\label{glvn2a1}
V_{n,2}(t) & = & - \frac{1}{2} \cdot \frac{1}{n} \sum_{k=1}^n \|X_k\|^2 {\rm CS}^{-}(t,X_k) t^\top \frac{1}{\sqrt{n}} \sum_{j=1}^n \left(X_jX_j^\top - {\rm I}_d\right) X_k \\ \nonumber
& & - \frac{1}{n} \sum_{k=1}^n \|X_k\|^2 {\rm CS}^{-}(t,X_k) \, \frac{1}{\sqrt{n}} \sum_{j=1}^n t^\top X_j + o_\PP(1).
\end{eqnarray}
In (\ref{glvn2a1}) we now use the fact that tr$(AB) = {\rm tr}(BA)$, where tr denotes trace and
 $AB$ is a square matrix. Furthermore, the error is $o_\PP(1)$ if we replace $n^{-1}\sum_{j=1}^n \|X_j\|^2{\rm CS}^{-}(t,X_j)$ and $n^{-1}\sum_{j=1}^n \|X_j\|^2 X_j {\rm CS}^{-}(t,X_j)$ with
 their almost sure limits
 $ \E[\|X\|^2 {\rm CS}^{-}(t,X)] = - \Delta \psi_X^{-}(t)$
 and
 $
 \E [\|X\|^2 X {\rm CS}^{-}(t,X)] = - \nabla \Delta \psi_X^{+}(t)$,
  respectively. We therefore obtain
 \begin{equation}\label{vn2final}
 V_{n,2}(t) = \frac{1}{\sqrt{n}} \sum_{j=1}^n \Big{\{} \frac{1}{2} \nabla \Delta \psi_X^{+}(t) \left(X_jX_j^\top - {\rm I}_d \right) t + \Delta \psi_X^{-}(t) t^\top X_j \Big{\}} + o_\PP(1).
 \end{equation}
In the same way, we proceed  with $V_{n,3}(t)$ and, using $\E [X {\rm CS}^{+}(t,X) X^\top] = - H\psi_X^{+}(t)$ as well as $\E [X{\rm CS}^{+}(t,X)] = - \nabla \psi_X^{-}(t)$,
finally arrive at
\begin{equation}\label{vn3final}
V_{n,3}(t) = \frac{1}{\sqrt{n}} \sum_{j=1}^n \Big{\{} 2 X_j^\top \nabla \psi_X^{-}(t) + X_j^\top H \psi_X^{-}(t) X_j + \mu(t)\Big{\}} + o_\PP(1).
\end{equation}
By adding (\ref{vn1final}), (\ref{vn2final}) and (\ref{vn3final}), we have $V_n = \widetilde{V}_n + o_\PP(1)$,
where $\widetilde{V}_n$ is given in (\ref{gldefvntild}).\hfill $\square$

\vspace*{1.5mm}

\noindent{\bf Proof of Theorem \ref{thmfixed1}}\\[1mm]
By analogy with $Z_n(t)$, as  defined in (\ref{defznt}), let
\[
Z^0_n(t) := \frac{1}{\sqrt{n}} \sum_{j=1}^n \big{\{} {\rm CS}^{+}(t,X_j) - m(t) \big{\}}, \quad t \in \R^d.
\]
A straightforward calculation gives $n^{-1/2}(Z_n(t) - Z^0_n(t)) = \sum_{j=1}^3 U_{n,j}(t)$,
where
\begin{eqnarray*}
U_{n,1}(t) & = & \frac{1}{n} \sum_{j=1}^n \|X_j\|^2 \left( {\rm CS}^{+}(t,Y_{n,j}) - {\rm CS}^{+}(t,X_j)  \right),\\
U_{n,2}(t) & = & \frac{2}{n} \sum_{j=1}^n X_j^\top \Delta_{n,j} {\rm CS}^{+}(t,Y_{n,j}), \quad
U_{n,3}(t)  =  \frac{1}{n} \sum_{j=1}^n \| \Delta_{n,j}\|^2 {\rm CS}^{+}(t,Y_{n,j}).
\end{eqnarray*}
Since $|\cos(t^\top Y_{n,j}) - \cos(t^\top X_j)| \le \|t\| \, \|\Delta_{n,j}\|$, $|\sin(t^\top Y_{n,j}) - \sin(t^\top X_j)| \le \|t\| \, \|\Delta_{n,j}\|$,
the Cauchy--Schwarz inequality gives $|U_{n,1}(t)| \le 2 \|t\| \big{(} n^{-1}\sum_{j=1}^n \|X_j\|^4\big{)}^{1/2} \big{(}n^{-1} \sum_{j=1}^n \|\Delta_{n,j}\|^2 \big{)}^{1/2}$.
By the strong law of large numbers, we have $n^{-1}\sum_{j=1}^n \|X_j\|^4 \fsk \E\|X\|^4$. In view of Proposition
\ref{propossumselta} b), we thus obtain $\|U_{n,1}\|^2_\HH \fsk 0$. Next, notice that, again by the Cauchy--Schwarz inequality,
$|U_{n,2}(t)| \le 4 \big{(} n^{-1}\sum_{j=1}^n \|X_j\|^2\big{)}^{1/2} \big{(}n^{-1} \sum_{j=1}^n \|\Delta_{n,j}\|^2 \big{)}^{1/2}$.
Hence, we have $\|U_{n,2}\|^2_\HH \fsk 0$. Finally, $|U_{n,3}(t)|  \le 2 n^{-1} \sum_{j=1}^n \|\Delta_{n,j}\|^2$ which, in view of Proposition
\ref{propossumselta} b), shows that $\|U_{n,3}\|^2_\HH \fsk 0$. Summarizing, we have
\begin{equation}\label{convznzn0}
\|n^{-1/2}(Z_n(\cdot ) - Z^0_n(\cdot ))\|_\HH \fsk 0.
\end{equation}
By the strong law of large numbers in Banach spaces, it follows that $n^{-1/2} Z^0_n(\cdot) \fsk \mu(\cdot) - m(\cdot)$ as $n \to \infty$
in $\HH$. In view of (\ref{convznzn0}), we thus obtain
\[
\frac{T_{n,a}}{n} = \| n^{-1/2} Z_n(\cdot )\|^2_\HH  \fsk \|\mu(\cdot) - m(\cdot)\|^2_\HH = \int \big{(}\mu(t)-m(t)\big{)}^2 w_a(t) \, {\rm d} t.
\]
\hfill $\square$

\noindent {\bf Proof of Theorem \ref{thmlimdeltaa}}\\[1mm]
Starting with (\ref{defdeltaa}), we have
\begin{eqnarray*}
\Delta_a & = & \int \left( \E (\|X\|^2 {\rm CS}^{+}(t,X)) \right)^2 \exp(-a\|t\|^2) \, {\rm d} t \\
& & - 2 \int (d-\|t\|)^2 \E (\|X\|^2 {\rm CS}^{+}(t,X)) \exp \left( - \frac{2a+1}{2} \|t\|^2 \right) {\rm d} t \\
& & + \int (d-\|t\|^2)^2 \exp(-(a+1)\|t\|^2)\, {\rm d} t\\
& =: & I_{1,a} -  I_{2,a} + I_{3,a},
\end{eqnarray*}
say. Letting $X_1,X_2$ be independent copies of $X$, Fubini's theorem, the addition theorems for the sine function and the cosine function,
considerations of symmetry and (\ref{glint1}) yield
\[
I_{1,a} = \left( \frac{\pi}{a}\right)^{d/2} \E \Big{[} \|X_1\|^2 \|X_2\|^2 \exp\left( - \frac{\|X_1-X_2\|^2}{4a}\right)\Big{]}.
\]
From (\ref{glint2}) and (\ref{glint3}), we have
\begin{eqnarray*}
I_{2,a} & = & \frac{2 (2\pi)^{d/2}}{(2a+1)^{d/2+2}} \,  \E \Big{[} \|X\|^2 \left( \|X\|^2 + 2da(2a+1)\right) \exp \left(- \frac{\|X\|^2}{2(2a+1)}\right) \Big{]},\\
I_{3,a} & = & \frac{\pi^{d/2}}{(a+1)^{d/2+2}} \left( a(a+1)d^2 + \frac{d(d+2)}{4}\right)
\end{eqnarray*}
and thus
\begin{eqnarray*}
2a \left(\frac{a}{\pi}\right)^{d/2} \Delta_a & = &  2a \, \E \Big{[} \|X_1\|^2 \|X_2\|^2 \exp\left( - \frac{\|X_1-X_2\|^2}{4a}\right)\Big{]}\\
& & -  \frac{2 (2a)^{d/2+1}}{(2a+1)^{d/2+2}} \E \Big{[} \|X\|^2 \left( \|X\|^2 + 2da(2a+1)\right) \exp \left(- \frac{\|X\|^2}{2(2a+1)}\right) \Big{]}\\
& & + \frac{2 a^{d/2+1}}{(a+1)^{d/2+2}} \left( a(a+1)d^2 + \frac{d(d+2)}{4}\right)\\
& =: & J_{1,a} - J_{2,a} + J_{3,a},
\end{eqnarray*}
say. An expansion of the exponential terms, dominated convergence (notice that $\exp(-u) \le 1-u + u^2$ if $u \ge 0$) and a binomial expansion gives
\begin{eqnarray*}
J_{1,a} & = & 2ad^2 - d \E \|X_1\|^4 + \E \left( \|X_1\|^2 \|X_2\|^2 X_1^\top X_2 \right) + O(a^{-1}),\\
J_{2,a} & = &  4 a d^2 - d^3 - 2 d^2 - d \E \|X_1\|^4 + O(a^{-1}),\\
J_{3,a} & = & 2ad^2 - d^3 - 2d^2 + O(a^{-1})
\end{eqnarray*}
and thus
\begin{eqnarray*}
\lim_{a \to \infty} 2a \left(\frac{a}{\pi}\right)^{d/2} \Delta_a & = & \E \left(\|X_1\|^2 \|X_2\|^2 X_1^\top X_2 \right) =
{\E \left(\|X_1\|^2 X_1\right)}^\top \E \left(\|X_2\|^2 X_2\right)\\
& = & \big{\|} \E (\|X\|^2 X)\big{\|}^2.
\end{eqnarray*}
Notice that the condition $\mathbb{E} \|X\|^6 < \infty$ is not only needed for the existence of the final limit, but it also
occurs when dealing with $J_{2,a}$.\hfill $\square$

\vspace*{1mm}

\noindent {\bf Proof of Lemma \ref{lemma6.6}}\\[1mm]
Putting $Y_j = Y_{n,j}$ and $\Delta_j = Y_j - X_j$, we have
\begin{align*}
    \cos(t^{\top}Y_j) &= \cos(t^{\top}X_j) - t^{\top}\Delta_j\sin(t^{\top}X_j) + \varepsilon_{j}(t), \\
    \sin(t^{\top}Y_j) &= \sin(t^{\top}X_j) + t^{\top}\Delta_j\cos(t^{\top}X_j) + \eta_{j}(t),
\end{align*}
where $|\varepsilon_{j}(t)|, |\eta_{j}(t)| \leq \|t\|^2\|\Delta_{j}\|^2$ (see~\cite{HW:1997}, p. 8) and thus
\begin{equation}\label{taylorcs+}
\text{CS}^{\pm}(t,Y_j) = \text{CS}^{\pm}(t, X_j) \pm t^\top \Delta_j \text{CS}^{\mp}(t, X_j) + \varepsilon_j(t) + \eta_j(t).
\end{equation}
To prove a), notice that $\Psi_{1,n}(t) = n^{-1} \sum_{j=1}^n \text{CS}^+(t, X_j)X_j + R_{1,n}(t)$, where
\begin{align*}
    R_{1,n}(t) &= \frac{1}{n}\sum_{j=1}^n \bigl(t^{\top}\Delta_{j}\text{CS}^-(t, X_j) + \varepsilon_j(t) + \eta_j(t)\bigr)X_j \\
    &+ \frac{1}{n}\sum_{j=1}^n \bigl(\text{CS}^+(t, X_j) + t^{\top}\Delta_{j}\text{CS}^-(t, X_j) + \varepsilon_j(t) + \eta_j(t)\bigr)\Delta_j.
\end{align*}
Since $|\text{CS}^{\pm}(t, X_j)| \leq 2$ and $|\varepsilon_{j}(t) + \eta_{j}(t)| \leq 2 \|t\|^2 \| \Delta_j\|^2$, the Cauchy--Schwarz inequality yields
\begin{align*}
    |R_{1,n}(t)| &\leq \frac{2\|t\|}{n}\sum_{j=1}^n \|\Delta_j\|\|X_j\| + \frac{2\|t\|^2}{n}\sum_{j=1}^n \|\Delta_j\|^2\|X_j\| + \frac{2}{n}\sum_{j=1}^n \|\Delta_j\| \\
    &+ \frac{2\|t\|}{n}\sum_{j=1}^n \|\Delta_j\|^2 + \frac{2\|t\|^2}{n}\sum_{j=1}^n \|\Delta_j\|^3.
\end{align*}
In view of Proposition~\ref{propossumseltaxx}, each summand converges to zero almost surely, which proves a).
The remaining assertions b), \ldots, e) are treated similarly. To tackle b) and e), one can show negligibility of terms
by using the fact that the spectral norm $\| \cdot \|_{{\rm sp}}$ of a matrix satisfies $\|xy^\top\|_{{\rm sp}} = \|x\|\, \|y\|$, if $x,y$ are
  (column) vectors in $\mathbb{R}^d$. The condition $\E \|X_1\|^6 < \infty$ is needed for part e), since
\begin{align*}
    \|Y_{j}\|^2\|Y_jY_j^{\top}\|_{\text{sp}} &\leq \left(\|X_j\|^2 + 2\|X_j\|\|\Delta_j\| + \|\Delta_j\|^2\right)^2 \\
    &= \|X_j\|^4 + 4\|X_j\|^3\|\Delta_j\| + 6\|X_j\|^2\|\Delta_j\|^2 + 4\|X_j\|\|\Delta_j\|^3 + \|\Delta_j\|^4,
\end{align*}
and multiplication with $\|\Delta_j\|^2$ (which arises from an expansion of ${\rm CS}^+(t, Y_{n,j})$) gives
 monomials of order 6.\hfill $\square$

\vspace*{1mm}
\noindent {\bf Proof of Theorem \ref{thmvarest}}\\[1mm]
The proof is similar to the proof of Theorem 5 of \cite{HM:2019} and will therefore only be sketched.
From (\ref{defkernl}),  (\ref{reprsumvj}) and  (\ref{darstsigma2}), we have $\sigma_a^2 = \sum_{i,j=1}^4 \sigma_a^{i,j}$, where
\[
\sigma_a^{i,j} = 4 \iint L^{i,j}(s,t)z(s)z(t) w_a(s)w_a(t) \, {\rm d}s {\rm d}t,
\]
and $L^{i,j}(s,t)$ is given in (\ref{deflijst}). It thus suffices to prove
$\widehat{\sigma}_{n,a}^{i,j} \stackrel{\mathbb{P}}{\longrightarrow} \sigma_a^{i,j}$ for each pair $(i,j)$, where
$\widehat{\sigma}_{n,a}^{i,j}$ is given in $(\ref{defsigmanij})$. The first step of the proof is to replace $L_n^{i,j}(s,t)$ with
$L_{n,0}^{i,j}(s,t)$, which arises from $L_n^{i,j}(s,t)$ given in (\ref{deflnij}) by throughout replacing $Y_{n,k}$ with $X_k$ in the functions $v_{n,j}(s,Y_{n,k})$, $j \in \{1,\ldots,4\}$.
Notice that this replacement also refers to the quantities $\Psi^{-}_{4,n}(s)$, $\Psi_{1,n}(s)$,  $\Psi^{-}_{3,n}(s)$ and $\Psi_{2,n}(s)$ that figure in the
definition of $v_{n,2}$, $v_{n,3}$ and $v_{n,4}$. Moreover, we replace $z_n(s)$ with $z_{n,0}(s) = n^{-1}\sum_{j=1}^n {\rm CS}^{+}(s,X_j)\|X_j\|^2 - m(s)$.
Putting
\[
\widehat{\sigma}_{n,0,a}^{i,j} = 4 \iint L_{n,0}^{i,j}(s,t)z_{n,0}(s)z_{n,0}(t) w_a(s)w_a(t) \, {\rm d}s {\rm d}t,
\]
Fubini's theorem shows that $\widehat{\sigma}_{n,0,a}^{i,j} \stackrel{\mathbb{P}}{\longrightarrow} \sigma_a^{i,j}$. It thus remains to prove
\[
\widehat{\sigma}_{n,a}^{i,j} - \widehat{\sigma}_{n,0,a}^{i,j}  = o_\mathbb{P}(1)
\]
as $n \to \infty$. To tackle $\widehat{\sigma}_{n,a}^{i,j} - \widehat{\sigma}_{n,0,a}^{i,j}$, we put
\[
\Lambda_n(s,t) = L_{n}^{i,j}(s,t)z_{n}(s)z_{n}(t) - L_{n,0}^{i,j}(s,t)z_{n,0}(s)z_{n,0}(t)
\]
and notice that
\begin{eqnarray}\label{lambda1}
\Lambda_n(s,t) & = & \left(L_{n}^{i,j}(s,t) - L_{n,0}^{i,j}(s,t)\right) z_n(s)z_n(t)\\ \label{lambda2}
& & \quad + L_{n,0}^{i,j}(s,t)\left(z_n(s)z_n(t)- z_{n,0}(s)z_{n,0}(t) \right), \\ \nonumber
 z_n(s)z_n(t)-z_{n,0}(s) z_{n,0}(t) & = &    \left(z_n(s)-z_{n,0}(s)\right) \left(z_n(t)-z_{n,0}(t)\right)\\ \nonumber
 & & \quad + z_{n,0}(s) \left(z_n(t)-z_{n,0}(t)\right) + z_{n,0}(t)\left(z_n(s)-z_{n,0}(s)\right).
\end{eqnarray}
From (\ref{defznsx}), we have $|z_n(s)| \le 2 d+ m(s)$. Moreover, $|z_{n,0}(s)| \le 2n^{-1}\sum_{j=1}^n \|X_j\|^2 + m(s)$.
A Taylor expansion shows that $|z_n(s) - z_{n,0}(s)|$ is bounded from above by a finite sum of terms of the type
$\|s\|^\ell n^{-1} \sum_{j=1}^n \|X_j\|^\beta \|\Delta_{n,j}\|^\gamma$, where $s \in \{0,1\}$, $\gamma \ge 1$ and $\beta + \gamma \le 3$.
Next, by straightforward calculations we obtain that $|L_{n,0}^{i,j}(s,t)|$ is bounded from above by $\|s\|^\ell \|t\|^m$ times finitely many products
of the type $n^{-1}\sum_{j=1}^n \|X_j\|^\beta$, where $\ell, m \in \{0,1\}$ and $\beta \in \{1,2,3,4\}$.
It now follows from Proposition \ref{propossumseltaxx} that the term figuring in (\ref{lambda2}), multiplied with $w_a(s)w_a(t)$ and integrated over $\mathbb{R}^d \times \mathbb{R}^d$,
converges to zero in probability.

As for the term figuring in (\ref{lambda1}), we have $|z_n(s)z_n(t)| \le (2d+m(s))(2d+m(t))$. To tackle
$L_{n}^{i,j}(s,t) - L_{n,0}^{i,j}(s,t)$, we confine ourselves to the case
$i=1, j=2$, since the other cases can be treated in a similar way. From (\ref{ln12}), we  have
\begin{eqnarray}\label{eqn9.11}
    L^{1,2}_n(s, t) \! & \! = \! & \!  -\frac{1}{2}\Bigg{\{}\frac{1}{n}\sum_{j=1}^n {\rm CS}^-(s,Y_j)\|Y_j\|^2Y_j\Bigg{\}}^\top\Bigg{\{}\frac{1}{n}\sum_{j=1}^n {\rm CS}^+(t, Y_j)\|Y_j\|^2Y_jY_j^\top s\Bigg{\}}\\ \label{eqn9.12}
   \! & \!  \! & \!  + \frac{1}{2}\Bigg{\{}\frac{1}{n}\sum_{j=1}^n {\rm CS}^+(t, Y_j)\|Y_j\|^2\Bigg{\}}\Bigg{\{}\frac{1}{n}\sum_{j=1}^n {\rm CS}^-(s,Y_j)\|Y_j\|^2Y_j^\top s\Bigg{\}}.
\end{eqnarray}
Moreover,
\begin{eqnarray}\label{eqn9.13}
       L_{n,0}^{1,2}(s,t) \! & \! = \! & \! -\frac{1}{2}\Bigg{\{}\frac{1}{n}\sum_{j=1}^n {\rm CS}^-(s, X_j)\|X_j\|^2X_j\Bigg{\}}^\top \! \! \Bigg{\{}\frac{1}{n}\sum_{j=1}^n {\rm CS}^+(t, X_j)\|X_j\|^2X_jX_j^\top s\Bigg{\}}\\ \label{eqn9.14}
    \! & \!  \! & \!  + \frac{1}{2}\Bigg{\{}\frac{1}{n}\sum_{j=1}^n {\rm CS}^+(t, X_j)\|X_j\|^2\Bigg{\}}\Bigg{\{}\frac{1}{n}\sum_{j=1}^n {\rm CS}^-(s,X_j)\|X_j\|^2X_j^\top s\Bigg{\}}.
    \end{eqnarray}
  Using (\ref{taylorcs+}), some calculations show that, for each $i \in \{1,2\}$, the norm of the difference of the $i^{{\rm th}}$ curly bracket in (\ref{eqn9.11}) and the   $i^{{\rm th}}$ curly bracket in (\ref{eqn9.13})
  is bounded from above by  finite sum of terms of the type
  \begin{equation}\label{prodtype}
  \|s\|^\ell \|t\|^m \, \frac{1}{n}\sum_{j=1}^n \|X_j\|^\beta \|\Delta_{n,j}\|^\gamma,
  \end{equation}
   where $\ell, m \in \{0,1,2\}$,
  $\gamma \ge 1$, and $\beta + \gamma \le 5$. Since the same holds for the norm of the difference of the $i^{{\rm th}}$ curly bracket in (\ref{eqn9.12}) and the   $i^{{\rm th}}$ curly bracket in (\ref{eqn9.14}), $i \in \{1,2\}$, it follows that $|L^{1,2}_n(s, t) - L_{n,0}^{1,2}(s,t)|$ is bounded from above by a finite sum of terms, which are either products of two terms of type (\ref{prodtype}), or a product of a term of type
  (\ref{prodtype}) and one of the terms inside one of the curly brackets in (\ref{eqn9.13}) or (\ref{eqn9.14}). From Proposition \ref{propossumseltaxx}, it follows that
  the term figuring in (\ref{lambda1}), multiplied with $w_a(s)w_a(t)$ and integrated over $\mathbb{R}^d \times \mathbb{R}^d$,
converges to zero in probability.\hfill $\square$

\end{appendix}
\begin{table}[b]
\setlength{\tabcolsep}{.5mm}
\centering
\footnotesize
\begin{tabular}{lc|cccccc|cccccc}
  \multicolumn{2}{c}{} & \multicolumn{6}{c}{$T_{n,a}$} & \multicolumn{6}{c}{} \\
 Alt. & $n\backslash a$ & 0.25 & 0.5 & 1 & 1.5 & 2 & 3 & BE & SW & BCMR & BHEP & AD & SF \\
  \hline
  \multirow{3}{*}{N$(0,1)$} & 20 & 5 & 5 & 5 & 5 & 5 & 5 & 5 & 5 & 5 & 5 & 5 & 5\\
  & 50 & 5 & 5 & 5 & 5 & 5 & 5 & 5 & 5 & 5 & 5 & 5 & 5 \\
  & 100 & 5 & 5 & 5 & 5 & 5 & 5 & 5 & 5 & 5 & 5 & 5 & 5 \\
  \hline
  \multirow{3}{*}{NMix$(0.3,1,0.25)$} & 20 & 13 & 15 & 18 & 19 & 19 & 19 & 23 & 28 & 28 & 27 & {\bf 30} & 25 \\
  & 50 & 34 & 38 & 44 & 46 & 46 & 46 & 54 & 60 & 60 & 62 & {\bf 68} & 57 \\
  & 100 & 67 & 71 & 77 & 78 & 78 & 76 & 85 & 89 & 89 & 90 & {\bf 94} & 88 \\
  \multirow{3}{*}{NMix$(0.5,1,4)$} & 20 & 38 & 39 & 40 & 39 & 38 & 36 & 36 & 40 & 43 & 42 & 46 & {\bf 48} \\
  & 50 & 67 & 69 & 70 & 68 & 65 & 61 & 64 & 78 & 80 & 80 & {\bf 86} & 83 \\
  & 100 & 92 & 94 & 94 & 93 & 91 & 87 & 92 & 97 & 98 & 98 & {\bf 99} & 98 \\
  \hline
  \multirow{3}{*}{t$_3$} & 20 & 39 & 39 & 39 & 38 & 38 & 37 & 34 &  35 & 37 & 34 & 33 & {\bf 40} \\
  & 50 & {\bf 69} & {\bf 69} & 68 & 67 & 66 & 64 & 57 & 64 & 65 & 61 & 60 & {\bf 69} \\
  & 100 & 89 & 89 & 89 & 88 & 87 & 86 & 81 & 88 & 89 & 86 & 85 & {\bf 91} \\
  \multirow{3}{*}{t$_5$} & 20 & {\bf 23} & {\bf 23} & 22 & 22 & 21 & 21 &20 & 19 & 20 & 18 & 17 & 22 \\
  & 50 & {\bf 41} & 40 & 39 & 38 & 37 & 36 & 31 & 35 & 37 & 32 & 31 & {\bf 41} \\
  & 100 & 61 & 61 & 61 & 59 & 58 & 55 & 46 & 56 & 58 & 50 & 48 & {\bf 63} \\
  \multirow{3}{*}{t$_{10}$} & 20 & {\bf 12} & {\bf 12} & {\bf 12} & {\bf 12} & {\bf 12} & 11 & 11 & 10 & 11 & 9 & 9 & {\bf 12}  \\
  & 50 & 19 & 19 & 19 & 18 & 17 & 17 & 14 & 16 & 17 & 13 & 12 & {\bf 20} \\
  & 100 & {\bf 28} & {\bf 28} & {\bf 28} & 27 & 26 & 25 & 19 & 22 & 24 & 16 & 15 & {\bf 28} \\
  \hline
  \multirow{3}{*}{U$(-\sqrt{3},\sqrt{3})$} & 20 & 1 & 1 & 1 & 1 & 1 & 1 & 3 &  {\bf 21} & 17 & 13 & 17 & 8 \\
  & 50 & 26 & 15 & 13 & 9 & 5 & 2 & 8 & {\bf 75} & 70 & 55 & 58 & 47 \\
  & 100 & 96 & 89 & 81 & 73 & 59 & 24 & 45 & {\bf $\ast$} & 99 & 95 & 95 & 97 \\
  \hline
  \multirow{3}{*}{$\chi^2_5$} & 20 & 35 & 37 & 39 & 41 & 41 & 42 & {\bf 44} & {\bf 44} & {\bf 44} & 42 & 38 & 42 \\
  & 50 & 79 & 81 & 83 & 84 & 85 & 85 & 87 & {\bf 89} & 88 & 84 & 80 & 85 \\
  & 100 & 99 & 99 & 99 & 99 & 99 & 99 & {\bf $\ast$} & {\bf $\ast$} & {\bf $\ast$} & 99 & 99 & {\bf $\ast$} \\
  \multirow{3}{*}{$\chi^2_{15}$} & 20 & 16 & 17 & 17 & 18 & 18 & 18 & {\bf 19} & 18 & 18 & 17 & 16 & 18 \\
  & 50 & 36 & 39 & 42 & 43 & 44 & 44 & {\bf 45} & 42 & 42 & 39 & 33 & 40 \\
  & 100 & 67 & 71 & 73 & 75 & 75 & 76 & {\bf 77} & 75 & 74 & 68 & 61 & 71 \\
  \hline
  \multirow{3}{*}{B$(1,4)$} & 20 & 34 & 38 & 42 & 44 & 45 & 45 &  50 & {\bf 59} & 58 & 52 & 51 & 53 \\
  & 50 & 90 & 91 & 92 & 93 & 93 & 92 & 94 & {\bf 98} & {\bf 98} & 94 & 95 & 97 \\
  & 100 & {\bf $\ast$} & {\bf $\ast$} & {\bf $\ast$} & {\bf $\ast$} & {\bf $\ast$} & {\bf $\ast$} & {\bf $\ast$} & {\bf $\ast$} & {\bf $\ast$} & {\bf $\ast$} & {\bf $\ast$} & {\bf $\ast$} \\
  \multirow{3}{*}{B$(2,5)$} & 20 & 11 & 12 & 13 & 14 & 15 & 15 & 15 & {\bf 16} & {\bf 16} & {\bf 16} & 14 & 14 \\
  & 50 & 31 & 33 & 37 & 39 & 40 & 40 & 43 & {\bf 50} & 47 & 45 & 39 & 40 \\
  & 100 & 78 & 79 & 79 & 80 & 80 & 79 & 81 & {\bf 90} & 89 & 80 & 76 & 82 \\
  \hline
  \multirow{3}{*}{$\Gamma(1,5)$} & 20 & 64 & 68 & 72 & 73 & 74 & 74 & 77 &  {\bf 83} & {\bf 83} & 77 & 77 & 80 \\
  & 50 & 99 & 99 & 99 & {\bf $\ast$} & {\bf $\ast$} & 99 & {\bf $\ast$} & {\bf $\ast$} & {\bf $\ast$} & {\bf $\ast$} & {\bf $\ast$} & {\bf $\ast$} \\
  & 100 & {\bf $\ast$} & {\bf $\ast$} & {\bf $\ast$} & {\bf $\ast$} & {\bf $\ast$} & {\bf $\ast$} & {\bf $\ast$} & {\bf $\ast$} & {\bf $\ast$} & {\bf $\ast$} & {\bf $\ast$} & {\bf $\ast$} \\
  \multirow{3}{*}{$\Gamma(5,1)$} & 20 & 21 & 22 & 23 & 24 & {\bf 25} & {\bf 25} & {\bf 25} & 24 & 24 & 23 & 20 & 24 \\
  & 50 & 50 & 54 & 57 & 59 & 59 & 60 & {\bf 61} & 59 & 59 & 55 & 49 & 56 \\
  & 100 & 84 & 87 & 88 & 89 & 89 & 90 & {\bf 91} & 90 & 90 & 85 & 81 & 88 \\
  \hline		
  \multirow{3}{*}{W$(1,0.5)$} & 20 & 65 & 68 & 72 & 74 & 75 & 74 & {\bf $\ast$} &  84 & 83 & 78 & 77 & 80 \\
  & 50 & 99 & 99 & 99 & {\bf $\ast$} & {\bf $\ast$} & {\bf $\ast$} & {\bf $\ast$} & {\bf $\ast$} & {\bf $\ast$} & {\bf $\ast$} & {\bf $\ast$} & {\bf $\ast$} \\
  & 100 & {\bf $\ast$} & {\bf $\ast$} & {\bf $\ast$} & {\bf $\ast$} & {\bf $\ast$} & {\bf $\ast$} & {\bf $\ast$} & {\bf $\ast$} & {\bf $\ast$} & {\bf $\ast$} & {\bf $\ast$} & {\bf $\ast$} \\
  \multirow{3}{*}{Gum$(1,2)$} & 20 & 28 & 30 & 32 & 32 & 33 & 33 & {\bf 34} & 31 & 32 & 31 & 28 & 32 \\
  & 50 & 63 & 67 & 69 & 70 & 71 & {\bf 72} & {\bf 72} & 69 & 69 & 66 & 60 & 67 \\
  & 100 & 92 & 94 & 94 & 95 & 95 & 95 & {\bf 96} & 94 & 94 & 91 & 89 & 93 \\
  \multirow{3}{*}{LN$(0,1)$} & 20 & 84 & 86 & 88 & 89 & 89 & 89 & 91 &  {\bf 93} & {\bf 93} & 91 & 90 & 91 \\
  & 50 &  {\bf $\ast$} & {\bf $\ast$} & {\bf $\ast$} & {\bf $\ast$} & {\bf $\ast$} & {\bf $\ast$} &  {\bf $\ast$} & {\bf $\ast$} & {\bf $\ast$} & {\bf $\ast$} & {\bf $\ast$} & {\bf $\ast$} \\
  & 100 &  {\bf $\ast$} & {\bf $\ast$} & {\bf $\ast$} & {\bf $\ast$} & {\bf $\ast$} & {\bf $\ast$} & {\bf $\ast$} & {\bf $\ast$} & {\bf $\ast$} & {\bf $\ast$} & {\bf $\ast$} & {\bf $\ast$}
\end{tabular}
\caption{Empirical power of $T_{n,a}$ against competitors ($d=1$, $\alpha = 0.05$, 10~000 replications)}\label{pow.T.1}
\end{table}

\begin{table}[b]
\setlength{\tabcolsep}{0.5mm}
\centering
\footnotesize
\begin{tabular}{lc|ccc|ccccccccc|ccccccccc}
 \multicolumn{5}{c}{}            & \multicolumn{9}{c}{BHEP$_a$} & \multicolumn{9}{c}{$T_{n,a}$} \\
 Alt. & $n$ & HV$_5$ & HV$_\infty$ & HJG$_{1.5}$  & 0.1 & 0.25 & 0.5 & 0.75 & 1 & 2 & 3 & 5 & 10 & 0.1 & 0.25 & 0.5 & 0.75 & 1 & 2 & 3 & 5 & 10\\
  \hline
  \multirow{2}{*}{NMix$(0.1,3,{\rm I}_2)$} & 20 & 32 & 37 & 33 & 38 & 39 & 41 & {\bf 42} & 38 & 23 & 17 & 11 & 8 & 31 & 33 & 34 & 35 & 35 & 38 & 39 & 38 & 37 \\
  & 50 & 69 & 85 & 73 & 86 &  87 & {\bf 88} & 86 & 83 & 62 & 47 & 29 & 16 & 73 & 80 & 82 & 83 & 84 & 86 & 86 & 85 & 83 \\

  \multirow{2}{*}{NMix$(0.5,0,{\rm B}_2)$} & 20 & {\bf 27} & 25 & {\bf 27} & 25 & 25 & 24 & 23 & 20 & 13 & 10 & 8 & 7 & {\bf 27} & {\bf 27} & {\bf 27} & {\bf 27} & {\bf 27} & 26 & 25 & 24 & 24 \\
  & 50& 53 & 43 & 52 & 43 & 45 & 47 & 43 & 37 & 23 & 17 & 12 & 9 & 54 & {\bf 55} & {\bf 55} & {\bf 55} & {\bf 55} & 53 & 50 & 48 & 45 \\

  \multirow{2}{*}{NMix$(0.9,0,{\rm B}_2)$} & {\bf 20} & 16 & 16 & 17 & 16 & 16 & 17 & 16 & 16 & 12 & 11 & 9 & 8 & {\bf 20} & 19 & 19 & 19 & 18 & 18 & 17 & 16 & 15 \\
  & 50& 25 & 21 & 24 & 21 & 21 & 26 & 30 & 32 & 28 & 23 & 16 & 11 & {\bf 40} & {\bf 40} & 37 & 35 & 34 & 30 & 27 & 23 & 21 \\
  \hline
  \multirow{2}{*}{$t_3(0,{\rm I}_2)$}& 20 & 54 & 53 & 54 & 52 & 52 & 52 & 51 & 48 & 37 & 31 & 24 & 15 & {\bf 58} & 57 & 57 & 57 & {\bf 58} & 57 & 55 & 53 & 51 \\
  & 50& 85 & 78 & 85 & 78 & 80 & 84 & 85 & 83 & 74 & 67 & 54 & 37 & {\bf 89} & 88 & 88 & {\bf 89} & {\bf 89} & {\bf 89} & 88 & 85 & 81 \\

  \multirow{2}{*}{$t_5(0,{\rm I}_2)$}& 20 & 33 & 31 & 33 & 31 & 31 & 30 & 28 & 25 & 16 & 13 & 10 & 8 & {\bf 35} & 34 & 34 & 34 & 34 & 34 & 32 & 31 & 30 \\
  & 50& 60 & 53 & 60 & 52 & 53 & 55 & 53 & 49 & 34 & 27 & 19 & 12 & {\bf 64} & 63 & 63 & {\bf 64} & {\bf 64} & {\bf 64} & 61 & 57 & 53 \\

  \multirow{2}{*}{$t_{10}(0,{\rm I}_2)$} & 20 & 16 & 16 & 16 & 15 & 15 & 14 & 13 & 11 & 8 & 7 & 6 & 6 & {\bf 17} & 16 & {\bf 17} & {\bf 17} & {\bf 17} & 16 & 16 & 15 & 15 \\
  & 50& 29 & 25 & 28 & 25 & 24 & 23 & 21 & 17 & 11 & 10 & 8 & 7 & {\bf 30} & 29 & 29 & {\bf 30} & {\bf 30} & {\bf 30} & 28 & 26 & 24 \\
  \hline
  \multirow{2}{*}{C$^2(0,1)$} & 20& 95 & 93 & 95 & 93 & 94 & 95 & 96 & {\bf 97} & 96 & 94 & 91 & 82 & {\bf 97} & {\bf 97} & {\bf 97} & {\bf 97} & {\bf 97} & {\bf 97} & 96 & 95 & 94 \\
  & 50& {\bf $\ast$} & 99 & {\bf $\ast$} & {\bf $\ast$} & {\bf $\ast$} & {\bf $\ast$} & {\bf $\ast$} & {\bf $\ast$} & {\bf $\ast$} & {\bf $\ast$} & {\bf $\ast$} & {\bf $\ast$} & {\bf $\ast$} & {\bf $\ast$} & {\bf $\ast$} & {\bf $\ast$} & {\bf $\ast$} & {\bf $\ast$} & {\bf $\ast$} & {\bf $\ast$} & {\bf $\ast$} \\

  \multirow{2}{*}{L$^2(0,1)$} & 20& 17 & 17 & 17 & 16 & 16 & 15 & 14 & 12 & 8 & 7 & 6 & 6 & 17 & 17 & 17 & {\bf 18} & {\bf 18} & 17 & 16 & 15 & 15 \\
  & 50& 29 & 24 & 29 & 24 & 24 & 24 & 21 & 18 & 12 & 10 & 8 & 7 & 31 & 31 & 31 & {\bf 32} & {\bf 32} & 30 & 28 & 26 & 24 \\

  \multirow{2}{*}{$\Gamma^2(0.5,1)$} & 20& 90 & 96 & 92 & 96 & 97 & 98 & {\bf 99} & {\bf 99} & 98 & 97 & 95 & 88 & 92 & 93 & 93 & 94 & 95 & 97 & 97 & 96 & 96 \\
  & 50& {\bf $\ast$} & {\bf $\ast$} & {\bf $\ast$} & {\bf $\ast$} & {\bf $\ast$} & {\bf $\ast$} & {\bf $\ast$} & {\bf $\ast$} & {\bf $\ast$} & {\bf $\ast$} & {\bf $\ast$} & {\bf $\ast$} & {\bf $\ast$} & {\bf $\ast$} & {\bf $\ast$} & {\bf $\ast$} & {\bf $\ast$} & {\bf $\ast$} & {\bf $\ast$} & {\bf $\ast$} & {\bf $\ast$} \\

  \multirow{2}{*}{$\Gamma^2(5,1)$} & 20& 22 & 25 & 23 & 26 & 27 & {\bf 28} & 27 & 25 & 16 & 12 & 9 & 7 & 20 & 21 & 22 & 22 & 23 & 25 & 26 & 26 & 26 \\
  & 50& 53 & 68 & 56 & 69 & 70 & {\bf 72} & 70 & 64 & 40 & 28 & 18 & 10 & 47 & 53 & 56 & 58 & 60 & 67 & 68 & 68 & 67 \\

  \multirow{2}{*}{P$_{VII}^2(10)$} & 20& 27 & 26 & 27 & 26 & 26 & 25 & 23 & 20 & 12 & 10 & 8 & 7 & {\bf 28} & {\bf 28} & {\bf 28} & {\bf 28} & {\bf 28} & {\bf 28} & 27 & 25 & 25 \\
  & 50& 50 & 44 & 50 & 43 & 44 & 44 & 42 & 38 & 25 & 20 & 15 & 10 & {\bf 55} & 54 & {\bf 55} & {\bf 55} & {\bf 55} & 54 & 51 & 47 & 44 \\

  \multirow{2}{*}{P$_{VII}^2(20)$} & 20& {\bf 14} & 13 & {\bf 14} & 13 & 13 & 12 & 11 & 10 & 7 & 6 & 6 & 5 & {\bf 14} & {\bf 14} & {\bf 14} & {\bf 14} & {\bf 14} & {\bf 14} & 13 & 13 & 13 \\
  & 50& 23 & 20 & 23 & 20 & 20 & 18 & 16 & 13 & 9 & 7 & 7 & 6 & {\bf 24} & {\bf 24} & {\bf 24} & {\bf 24} & {\bf 24} & {\bf 24} & 22 & 21 & 20 \\
  \hline
  \multirow{2}{*}{$\mathcal{S}^2(\mbox{Exp}(1))$} & 20& 67 & 66 & 67 & 65 & 66 & 70 & 74 & 77 & {\bf 79} & 78 & 72 & 59 & 77 & 75 & 75 & 76 & 77 & 76 & 74 & 69 & 64 \\
  & 50& 92 & 83 & 92 & 83 & 89 & 96 & 99 & 99 & {\bf $\ast$} & {\bf $\ast$} & 99 & 97 & 98 & 98 & 98 & 99 & 99 & 99 & 98 & 96 & 91 \\

  \multirow{2}{*}{$\mathcal{S}^2({\rm B}(1,2))$}& 20& 14 & 15 & 14 & 14 & 15 & 17 & 20 & 24 & 33 & {\bf 36} & 35 & 28 & 24 & 22 & 21 & 21 & 22 & 21 & 18 & 16 & 14 \\
  & 50& 10 & 12 & 10 & 11 & 13 & 23 & 38 & 52 & 75 & {\bf 78} & 76 & 66 & 34 & 32 & 32 & 33 & 34 & 32 & 26 & 20 & 14 \\

  \multirow{2}{*}{$\mathcal{S}^2(\chi^2_5)$}& 20& 22 & 21 & 22 & 21 & 21 & 20 & 18 & 16 & 10 & 9 & 8 & 7 & {\bf 23} & {\bf 23} & {\bf 23} & {\bf 23} & {\bf 23} & {\bf 23} & 22 & 21 & 20 \\
  & 50& 38 & 32 & 38 & 32 & 32 & 33 & 32 & 29 & 20 & 15 & 12 & 9 & 43 & 42 & 42 & 43 & {\bf 44} & 43 & 40 & 36 & 32 \\

\end{tabular}
\caption{Empirical power of $T_{n,a}$ against competitors ($d=2$, $\alpha = 0.05$, 10~000 replications)}\label{pow.T.2}
\end{table}

\begin{table}[b]
\setlength{\tabcolsep}{0.5mm}
\centering
\footnotesize
\begin{tabular}{lc|ccc|ccccccccc|ccccccccc}
 \multicolumn{5}{c}{}            & \multicolumn{9}{c}{BHEP$_a$} & \multicolumn{9}{c}{$T_{n,a}$} \\
 Alt. & $n$ & HV$_5$ & HV$_\infty$ & HJG$_{1.5}$  & 0.1 & 0.25 & 0.5 & 0.75 & 1 & 2 & 3 & 5 & 10 & 0.1 & 0.25 & 0.5 & 0.75 & 1 & 2 & 3 & 5 & 10\\
  \hline
  \multirow{2}{*}{NMix$(0.1,3,{\rm I}_3)$} & 20 & 33 & 39 & 34 & 39 & 41 & {\bf 44} & {\bf 44} & 38 & 20 & 14 & 9 & 7 & 31 & 34 & 35 & 35 & 35 & 37 & 39 & 39 & 38 \\
  & 50 & 67 & 92 & 72 & 91 & 93 & {\bf 95} & 93 & 89 & 64 & 42 & 21 & 10 & 71 & 85 & 87 & 86 & 86 & 90 & 91 & 90 & 88 \\

  \multirow{2}{*}{NMix$(0.5,0,{\rm B}_3)$} & 20 & {\bf 43} & 42 & {\bf 43} & 41 & 41 & 39 & 35 & 28 & 16 & 12 & 9 & 7 & 41 & 42 & 42 & 42 & 42 & {\bf 43} & 42 & 40 & 38 \\
  & 50 & 80 & 72 & 79 & 72 & 73 & 74 & 68 & 60 & 36 & 25 & 15 & 9 & 80 & 80 &{\bf  81} & {\bf 81} & {\bf 81} & 80 & 79 & 76 & 73 \\

  \multirow{2}{*}{NMix$(0.9,0,{\rm B}_3)$} & 20 & 28 & 29 & 28 & 28 & 28 & 29 & 30 & 28 & 22 & 17 & 12 & 9 & {\bf 37} & {\bf 37} & 35 & 34 & 33 & 32 & 30 & 27 & 25 \\
  & 50 & 45 & 40 & 44 & 38 & 42 & 53 & 63 & 67 & 62 & 49 & 28 & 13 & 75 & {\bf 76} & 72 & 68 & 67 & 62 & 58 & 49 & 40 \\
  \hline
  \multirow{2}{*}{$t_3(0,{\rm I}_3)$} & 20 & 65 & 66 & 65 & 64 & 64 & 64 & 61 & 56 & 41 & 32 & 21 & 13 & {\bf 71} & 70 & 69 & 69 & 69 & 69 & 67 & 64 & 61 \\
  & 50 & 94 & 91 & 94 & 90 & 92 & 94 & 94 & 93 & 85 & 75 & 56 & 28 & {\bf 98} & 97 & 97 & 97 & 97 & 97 & 96 & 95 & 92 \\

  \multirow{2}{*}{$t_5(0,{\rm I}_3)$} & 20 & 40 & 40 & 40 & 39 & 39 & 38 & 34 & 29 & 17 & 13 & 10 & 8 & {\bf 45} & 44 & 43 & 43 & 43 & 43 & 42 & 38 & 36 \\
  & 50 & 72 & 68 & 72 & 67 & 68 & 69 & 66 & 61 & 42 & 32 & 20 & 11 & {\bf 80} & 78 & 78 & 78 & 78 & 79 & 77 & 73 & 67 \\

  \multirow{2}{*}{$t_{10}(0,{\rm I}_3)$} & 20 & 20 & 20 & 20 & 19 & 19 & 18 & 15 & 12 & 8 & 7 & 7 & 6 & {\bf 21} & {\bf 21} & {\bf 21} & {\bf 21} & {\bf 21} & {\bf 21} & 20 & 18 & 18 \\
  & 50 & 39 & 35 & 39 & 34 & 33 & 32 & 27 & 22 & 13 & 10 & 8 & 6 & {\bf 45} & 42 & 42 & 42 & 43 & 43 & 41 & 36 & 33 \\
  \hline
  \multirow{2}{*}{C$^3(0,1)$} & 20 & 98 & 98 & 98 & 97 & 98 & 98 & {\bf 99} & {\bf 99} & 98 & 96 & 91 & 77 & {\bf 99} & {\bf 99} & {\bf 99} & {\bf 99} & {\bf 99} & {\bf 99} & {\bf 99} & 98 & 97 \\
  & 50 & {\bf $\ast$} & {\bf $\ast$} & {\bf $\ast$} & {\bf $\ast$} & {\bf $\ast$} & {\bf $\ast$} & {\bf $\ast$} & {\bf $\ast$} & {\bf $\ast$} & {\bf $\ast$} & {\bf $\ast$} & {\bf $\ast$} & {\bf $\ast$} & {\bf $\ast$} & {\bf $\ast$} & {\bf $\ast$} & {\bf $\ast$} & {\bf $\ast$} & {\bf $\ast$} & {\bf $\ast$} & {\bf $\ast$} \\

  \multirow{2}{*}{L$^3(0,1)$} & 20 & 16 & 16 & 16 & 16 & 16 & 15 & 13 & 11 & 7 & 6 & 6 & 6 & {\bf 18} & {\bf 18} & {\bf 18} & {\bf 18} & {\bf 18} & 17 & 16 & 16 & 15 \\
  & 50 & 31 & 29 & 31 & 27 & 27 & 26 & 22 & 18 & 11 & 10 & 8 & 6 & {\bf 39} & 37 & 36 & 36 & 36 & 36 & 34 & 29 & 26 \\

  \multirow{2}{*}{$\Gamma^3(0.5,1)$} & 20 & 92 & 98 & 93 & 98 & 98 & 99 & {\bf $\ast$} & {\bf $\ast$} & 98 & 96 & 89 & 68 & 94 & 96 & 95 & 95 & 95 & 97 & 98 & 97 & 97 \\
  & 50 & {\bf $\ast$} & {\bf $\ast$} & {\bf $\ast$} & {\bf $\ast$} & {\bf $\ast$} & {\bf $\ast$} & {\bf $\ast$} & {\bf $\ast$} & {\bf $\ast$} & {\bf $\ast$} & {\bf $\ast$} & {\bf $\ast$} & {\bf $\ast$} & {\bf $\ast$} & {\bf $\ast$} & {\bf $\ast$}
   & {\bf $\ast$} & {\bf $\ast$} & {\bf $\ast$} & {\bf $\ast$} & {\bf $\ast$} \\

  \multirow{2}{*}{$\Gamma^3(5,1)$} & 20 & 22 & 26 & 23 & 27 & 28 & {\bf 29} & 28 & 24 & 13 & 10 & 8 & 6 & 18 & 20 & 21 & 22 & 22 & 24 & 26 & 27 & 26 \\
  & 50 & 54 & 74 & 58 & 75 & 77 & {\bf 78} & 74 & 65 & 36 & 23 & 13 & 8 & 47 & 56 & 60 & 62 & 63 & 69 & 72 & 72 & 71 \\

  \multirow{2}{*}{P$_{VII}^3(10)$} & 20 & 31 & 30 & 30 & 29 & 29 & 28 & 24 & 19 & 11 & 9 & 7 & 7 & {\bf 33} & 32 & 32 & 32 & 32 & 32 & 31 & 28 & 27 \\
  & 50 & 59 & 54 & 59 & 52 & 53 & 53 & 48 & 41 & 26 & 18 & 12 & 8 & {\bf 67} & 66 & 65 & 65 & 65 & 64 & 62 & 57 & 52 \\

  \multirow{2}{*}{P$_{VII}^3(20)$} & 20 & 14 & 14 & 14 & 13 & 13 & 12 & 11 & 9 & 6 & 6 & 6 & 6 & {\bf 15} & {\bf 15} & {\bf 15} & {\bf 15} & {\bf 15} & 14 & 14 & 13 & 12 \\
  & 50 & 26 & 24 & 26 & 23 & 22 & 21 & 17 & 13 & 8 & 8 & 7 & 6 & 29 & 27 & {\bf 28} & {\bf 28} & {\bf 28} & {\bf 28} & 26 & 24 & 22 \\
  \hline
  \multirow{2}{*}{$\mathcal{S}^3(\mbox{Exp}(1))$} & 20 & 89 & 89 & 88 & 87 & 89 & 91 & 94 & 95 & 95 & 94 & 90 & 76 & {\bf 96} & 95 & 94 & 94 & 94 & 95 & 94 & 91 & 86 \\
  & 50 & 99 & 98 & 99 & 98 & 99 & {\bf $\ast$} & {\bf $\ast$} & {\bf $\ast$} & {\bf $\ast$} & {\bf $\ast$} & {\bf $\ast$} & {\bf $\ast$} & {\bf $\ast$} & {\bf $\ast$} & {\bf $\ast$} & {\bf $\ast$} & {\bf $\ast$} & {\bf $\ast$} & {\bf $\ast$} & {\bf $\ast$} & 99 \\

  \multirow{2}{*}{$\mathcal{S}^3({\rm B}(1,2))$} & 20 & 45 & 49 & 44 & 47 & 48 & 53 & 60 & 65 & {\bf 73} & {\bf 73} & 65 & 49 & 68 & 64 & 61 & 61 & 62 & 62 & 58 & 51 & 44 \\
  & 50 & 56 & 54 & 56 & 50 & 59 & 82 & 93 & 97 & {\bf 99} & {\bf 99} & 98 & 92 & 94 & 91 & 90 & 91 & 92 & 93 & 90 & 81 & 63 \\

  \multirow{2}{*}{$\mathcal{S}^3(\chi^2_5)$} & 20 & 43 & 44 & 42 & 43 & 43 & 43 & 41 & 38 & 29 & 23 & 16 & 11 & {\bf 54} & 52 & 50 & 50 & 50 & 50 & 48 & 43 & 40 \\
  & 50 & 73 & 68 & 73 & 66 & 69 & 76 & 79 & 79 & 73 & 62 & 44 & 20 & {\bf 89} & 86 & 85 & 85 & 86 & 87 & 85 & 79 & 69 \\
 \end{tabular}
\caption{Empirical power of $T_{n,a}$ against competitors ($d=3$, $\alpha = 0.05$, 10~000 replications)}\label{pow.T.3}
\end{table}

\begin{table}[t]
\setlength{\tabcolsep}{0.5mm}
\centering
\footnotesize
\begin{tabular}{lc|ccc|ccccccccc|ccccccccc}
 \multicolumn{5}{c}{}            & \multicolumn{9}{c|}{BHEP$_a$} & \multicolumn{9}{c}{$T_{n,a}$} \\
 Alt. & $n$ & HV$_5$ & HV$_\infty$ & HJG$_{1.5}$  & 0.1 & 0.25 & 0.5 & 0.75 & 1 & 2 & 3 & 5 & 10 & 0.1 & 0.25 & 0.5 & 0.75 & 1 & 2 & 3 & 5 & 10\\
  \hline
  \multirow{2}{*}{NMix$(0.1,3,{\rm I}_5)$} & 20  & 31 & 34 & 32 & {\bf 35} & {\bf 35} & {\bf 35} & 31 & 24 & 12 & 9 & 8 & 3 & 26 & 27 & 29 & 30 & 30 & 31 & 33 & 34 & 33 \\
  & 50 & 51 & 86 & 58 & 87 & 90 & {\bf 95} & 93 & 85 & 38 & 17 & 10 & 8 & 48 & 65 & 77 & 78 & 76 & 75 & 81 & 85 & 83 \\

  \multirow{2}{*}{NMix$(0.5,0,{\rm B}_5)$}& 20  & {\bf 62} & 60 & {\bf 62} & 59 & 58 & 54 & 43 & 31 & 14 & 10 & 9 & 3 & 56 & 57 & 58 & 59 & 59 & 60 & 60 & 58 & 55 \\
  & 50 & {\bf 96} & 94 & {\bf 96} & 94 & 94 & 92 & 86 & 76 & 43 & 24 & 13 & 10 & {\bf 96} & 95 & 95 & 95 & 95 & {\bf 96} & {\bf 96} & 95 & 94 \\

  \multirow{2}{*}{NMix$(0.9,0,{\rm B}_5)$}& 20 & 55 & 59 & 54 & 57 & 57 & 59 & 59 & 55 & 37 & 24 & 18 & 4 & 71 & {\bf 72} & 70 & 68 & 66 & 65 & 63 & 58 & 51 \\
  & 50 & 78 & 78 & 78 & 74 & 81 & 93 & 97 & 99 & 95 & 77 & 35 & 20 & 98 & {\bf 99} & 98 & 98 & 97 & 96 & 95 & 92 & 81 \\
  \hline
  \multirow{2}{*}{$t_3(0,{\rm I}_5)$}& 20 & 80 & 81 & 79 & 79 & 79 & 77 & 70 & 62 & 39 & 26 & 20 & 5 & {\bf 85} & 84 & 83 & 83 & 83 & 83 & 82 & 79 & 74 \\
  & 50 & 99 & 99 & 99 & 98 & 99 & 99 & 99 & 98 & 92 & 76 & 42 & 24 & {\bf $\ast$} & {\bf $\ast$} & {\bf $\ast$} & {\bf $\ast$} & {\bf $\ast$} & {\bf $\ast$} & {\bf $\ast$} & {\bf $\ast$} & 99 \\

  \multirow{2}{*}{$t_5(0,{\rm I}_5)$}& 20 & 54 & 55 & 54 & 53 & 53 & 49 & 40 & 30 & 15 & 11 & 10 & 3 & {\bf 59} & {\bf 59} & 58 & 58 & 58 & 57 & 57 & 53 & 48 \\
  & 50 & 89 & 89 & 89 & 87 & 87 & 88 & 84 & 78 & 52 & 31 & 15 & 11 & {\bf 96} & {\bf 96} & 95 & 94 & 94 & 95 & 94 & 92 & 87 \\

  \multirow{2}{*}{$t_{10}(0,{\rm I}_5)$}& 20 & 25 & 25 & 24 & 24 & 24 & 22 & 17 & 12 & 8 & 7 & 7 & 2 & {\bf 28} & {\bf 28} & {\bf 28} & {\bf 28} & 27 & 27 & 26 & 24 & 22 \\
  & 50 & 55 & 54 & 55 & 51 & 51 & 47 & 38 & 29 & 14 & 10 & 8 & 7 & {\bf 67} & 65 & 63 & 62 & 63 & 63 & 62 & 58 & 49 \\
  \hline
  \multirow{2}{*}{C$^5(0,1)$}& 20 & {\bf $\ast$} & {\bf $\ast$} & {\bf $\ast$} & 99 & {\bf $\ast$} & {\bf $\ast$} & {\bf $\ast$} & 99 & 98 & 94 & 88 & 61 & {\bf $\ast$} & {\bf $\ast$} & {\bf $\ast$} & {\bf $\ast$} & {\bf $\ast$} & {\bf $\ast$} & {\bf $\ast$} & {\bf $\ast$} & 99 \\
  & 50 & {\bf $\ast$} & {\bf $\ast$} & {\bf $\ast$} & {\bf $\ast$} & {\bf $\ast$} & {\bf $\ast$} & {\bf $\ast$} & {\bf $\ast$} & {\bf $\ast$} & {\bf $\ast$} & {\bf $\ast$} & {\bf $\ast$} & {\bf $\ast$} & {\bf $\ast$} & {\bf $\ast$} & {\bf $\ast$} & {\bf $\ast$} & {\bf $\ast$} & {\bf $\ast$} & {\bf $\ast$} & {\bf $\ast$} \\

  \multirow{2}{*}{L$^5(0,1)$}& 20 & 17 & 18 & 17 & 17 & 17 & 15 & 12 & 9 & 6 & 6 & 6 & 3 & {\bf 19} & {\bf 19} & {\bf 19} & {\bf 19} & {\bf 19} & {\bf 19} & 18 & 17 & 15 \\
  & 50 & 36 & 35 & 35 & 32 & 32 & 29 & 22 & 16 & 9 & 8 & 7 & 6 & {\bf 46} & 45 & 43 & 43 & 43 & 43 & 42 & 37 & 30 \\

  \multirow{2}{*}{$\Gamma^5(0.5,1)$}& 20 & 93 & 98 & 94 & 98 & 99 & 99 & {\bf $\ast$} & 99 & 93 & 81 & 69 & 34 & 93 & 96 & 97 & 96 & 96 & 96 & 97 & 98 & 97 \\
  & 50 & {\bf $\ast$} & {\bf $\ast$} & {\bf $\ast$} & {\bf $\ast$} & {\bf $\ast$} & {\bf $\ast$} & {\bf $\ast$} & {\bf $\ast$} & 100 & {\bf $\ast$} & 99 & 91 & {\bf $\ast$} & {\bf $\ast$} & {\bf $\ast$} & {\bf $\ast$} & {\bf $\ast$} & {\bf $\ast$} & {\bf $\ast$} & {\bf $\ast$} & {\bf $\ast$} \\

  \multirow{2}{*}{$\Gamma^5(5,1)$}& 20 & 19 & 23 & 21 & 24 & 24 & {\bf 25} & 23 & 18 & 9 & 7 & 7 & 4 & 17 & 18 & 19 & 19 & 19 & 20 & 23 & 24 & 24 \\
  & 50 & 50 & 75 & 55 & 76 & 78 & 80 & 73 & 59 & 23 & 12 & 7 & 7 & 44 & 51 & 59 & 61 & 61 & 63 & 68 & {\bf 73} & 72 \\

  \multirow{2}{*}{P$_{VII}^5(10)$}& 20 & 32 & 32 & 32 & 31 & 30 & 28 & 21 & 15 & 9 & 8 & 7 & 2 & {\bf 35} & 34 & 34 & 34 & 34 & 34 & 33 & 30 & 27 \\
  & 50 & 67 & 65 & 67 & 63 & 62 & 59 & 50 & 40 & 20 & 13 & 9 & 8 & {\bf 77} & 76 & 75 & 75 & 75 & 74 & 73 & 69 & 60 \\

  \multirow{2}{*}{P$_{VII}^5(20)$}& 20 & 13 & {\bf 14} & 13 & 13 & 12 & 11 & 10 & 7 & 6 & 6 & 6 & 3 & {\bf 14} & {\bf 14} & {\bf 14} & {\bf 14} & {\bf 14} & {\bf 14} & {\bf 14} & 13 & 11 \\
  & 50 & 28 & 27 & 28 & 26 & 25 & 22 & 16 & 12 & 7 & 6 & 6 & 6 & 34 & 33 & 31 & {\bf 32} & {\bf 32} & {\bf 32} & 31 & 28 & 24 \\
  \hline
  \multirow{2}{*}{$\mathcal{S}^5(\mbox{Exp}(1))$}& 20 & 99 & 99 & 99 & 99 & 99 & 99 & 99 & 99 & 99 & 98 & 96 & 83 & {\bf $\ast$} & {\bf $\ast$} & {\bf $\ast$} & {\bf $\ast$} & {\bf $\ast$} & {\bf $\ast$} & {\bf $\ast$} & 99 & 98 \\
  & 50 & {\bf $\ast$} & {\bf $\ast$} & {\bf $\ast$} & {\bf $\ast$} & {\bf $\ast$} & {\bf $\ast$} & {\bf $\ast$} & {\bf $\ast$} & {\bf $\ast$} & {\bf $\ast$} & {\bf $\ast$} & {\bf $\ast$}
   & {\bf $\ast$} & {\bf $\ast$} & {\bf $\ast$}
    & {\bf $\ast$} & {\bf $\ast$} & {\bf $\ast$} & {\bf $\ast$} & {\bf $\ast$} & {\bf $\ast$} \\

  \multirow{2}{*}{$\mathcal{S}^5({\rm B}(1,2))$}& 20 & 86 & 91 & 86 & 88 & 89 & 91 & 93 & 94 & 94 & 91 & 86 & 66 & {\bf 97} & {\bf 97} & 96 & 96 & 95 & 95 & 94 & 91 & 84 \\
  & 50 & 98 & 98 & 99 & 97 & 99 & {\bf $\ast$} & {\bf $\ast$} & {\bf $\ast$} & {\bf $\ast$} & {\bf $\ast$} & {\bf $\ast$} & 99 & {\bf $\ast$} & {\bf $\ast$} & {\bf $\ast$} & {\bf $\ast$} & {\bf $\ast$} & {\bf $\ast$} & {\bf $\ast$} & {\bf $\ast$} & 99 \\

  \multirow{2}{*}{$\mathcal{S}^5(\chi^2_5)$}& 20 & 77 & 80 & 76 & 78 & 78 & 77 & 74 & 71 & 59 & 46 & 36 & 12 & {\bf 89} & {\bf 89} & 87 & 86 & 86 & 85 & 84 & 79 & 71 \\
  & 50 & 98 & 98 & 98 & 96 & 98 & 99 & {\bf $\ast$} & {\bf $\ast$} & 99 & 95 & 74 & 50 & {\bf $\ast$} & {\bf $\ast$} & {\bf $\ast$} & {\bf $\ast$} & {\bf $\ast$} & {\bf $\ast$} & {\bf $\ast$} & {\bf $\ast$} & 98 \\
   \hline
\end{tabular}
\caption{Empirical power of $T_{n,a}$ against competitors ($d=5$, $\alpha = 0.05$, 10~000 replications)}\label{pow.T.5}
\end{table}

\bibliographystyle{unsrt}


\end{document}